\documentclass{elsarticle}
\usepackage{amsmath,amsthm,amssymb,graphics,pstool,wrapfig,url}
\usepackage[margin=0.9in]{geometry}
\usepackage{times}
\usepackage{multirow}
\usepackage{caption}

\usepackage[utf8]{inputenc}

\usepackage{lineno,hyperref}
\modulolinenumbers[5]

\journal{Journal of Infectious Disease Modeling}

\newtheorem{theorem}{Theorem}[section]
\newtheorem{proposition}[theorem]{Proposition}
\newtheorem{lemma}[theorem]{Lemma}

\begin{document}


\begin{frontmatter}
\title{The role of super-spreaders in modeling of SARS-CoV-2}



\author{Fran\c cois Rousse$^1$}
\author{Marcus Carlsson$^{2}$}
\author{Magnus \"{O}gren$^{1,3,*}\corref{mycorrespondingauthor}$}
\cortext[mycorrespondingauthor]{Corresponding author: magnus.ogren@oru.se}
\author{Benjamin Kalischer Wellander$^4$}
\address{$^1$School of Science and Technology, \"{O}rebro University, 70182 \"{O}rebro, Sweden}
\address{$^2$Centre for Mathematical Sciences, Lund University, Box 118, 22100 Lund, Sweden}
\address{$^3$Hellenic Mediterranean University, P.O. Box 1939, GR-71004, Heraklion, Greece}
\address{$^4$Centrum f\"{o}r forskning och utveckling (CFUG) Region G\"{a}vleborg, Sweden}

%
%

\begin{abstract}
In stochastic modeling of infectious diseases, it has been established that variations in infectivity affect the probability of a major outbreak, but not the shape of the curves during a major outbreak, which is predicted by deterministic models~\cite{diekmann2012mathematical}. However, such conclusions are derived under idealized assumptions such as the population size tending to infinity, and the individual degree of infectivity only depending on variations in the infectiousness period. In this paper we show that the same conclusions hold true in a finite population representing a medium size city, where the degree of infectivity is determined by the offspring distribution, which we try to make as realistic as possible for SARS-CoV-2. In particular, we consider distributions with fat tails, to incorporate the existence of super-spreaders. We also provide new theoretical results on convergence of stochastic models which allows to incorporate any offspring distribution with a finite variance.
\end{abstract}



\end{frontmatter}



\section{Introduction}

\textcolor{black}{
The spread of SARS-CoV-2 has been notoriously difficult to predict using mathematical models for the spread of infectious diseases. These models, such as SIR and SEIR, go back to the 1920's and the seminal works of Kermack and McKendrick \cite{kermack1927contribution}. For a virus with the characteristics of SARS-CoV-2, these models predict that at least 60\% of the population will become infected before the epidemic calms down, which happens within a couple of months. However, after the first wave had receded, most countries of Europe had a sero-prevalence below 20\%. This applies in particular to Sweden, which did comparatively little to halt transmission, making it a good example for comparison of model outcome and reality. This is studied in~\cite{carlsson2021indications}, which finds that the  discrepancy is so large that it raises the question of whether deterministic models (such as SIR and SEIR) properly capture reality, or whether some characteristic of SARS-CoV-2 is at odds with standard assumptions underlying these models.}

\textcolor{black}{A potential such feature is the existence of super-spreaders, or more generally, a very uneven distribution of individual levels of infectiousness \cite{adam2020clustering,endo2020estimating,jones2021estimating,wong2020evidence}}. While the law of large numbers intuitively implies that this effect should even out during large outbreaks, prior studies related to this topic do not prove that this is the case (see Section \ref{priorwork} for a brief review of known results).
In this paper we investigate how much a stochastic model taking variable infectiousness into account is likely to differ from a standard deterministic SEIR-model that assumes a homogenous distribution of infectiousness.

\textcolor{black}{More precisely, we prove theoretically that our stochastic model converges to the deterministic SEIR when the population size $N$ approaches infinity (see Appendix \ref{convergence}). Similar statements for homogenous models are well known, see e.g.~\cite{diekmann2012mathematical}, but these results are derived under the unrealistic assumption that all individuals are equally infective. Our result is new since it does not have this limitation, and also the method of proof differs significantly from prior works.} However, such probabilistic statements do not imply that the stochastic model will behave like a deterministic SEIR when modeling with a fixed population size of, say, one million (i.e.~an average large city).

\textcolor{black}{The second main contribution of the present paper is a numerical investigation of the impact of infectivity
variations with a focus on the probability of outbreaks and the
shape of the outbreak curves.  In particular we test four
types of distributions for infectivity variations, using two different Agent-Based models in a population of one million.}

\textcolor{black}{In order do this realistically, we need the offspring distribution for SARS-CoV-2, which unfortunately is hard to estimate. However, we find that our conclusions are fairly robust with respect to which type of distributions we use, and align with prior theoretical results based on simpler models, as found e.g. in \cite{diekmann2012mathematical}. In particular, we find that the probability of a major outbreak decrease when the offspring distribution has a ``fat-tail'', but that once a major outbreak is happening, the shape is well predicted by deterministic SEIR, irrespective of the distribution used.}

\section{Methods}\label{novelties}

By ``secondary cases'' we shall refer to the amount of new infections a given infected individual gives rise to. The corresponding distribution is known as the offspring distribution, and a standard assumption is that this distribution is negative binomial with a small value of the dispersion parameter $k$ \cite{blumberg2013inference,lloyd2005superspreading}. Small values of $k$ are (somewhat counterintuitively) associated with a high variance and hence that super spreaders are common. It has also been suggested that the offspring distribution could be fat tailed \cite{wong2020evidence}, which would give rise to even more irregular behavior also during large outbreaks \cite{taleb2020single}.

The shape of the offspring distribution was systematically investigated \cite{kremer2021quantifying}.
While the authors did not find support for the hypothesis of a fat-tailed shape in a strict sense, they do provide various estimates indicating that a minority of infected persons were responsible for the majority of transmissions. One such indicator is the estimated percentage of spreaders that are accountable for 80\% of secondary cases, denoted $p_{80}$. For example, using contact tracing data from Hong-Kong, they estimate that $p_{80} = 19\%$. This seems to align with the results of other observational studies that collectively estimate $p_{80}$ in the range $ 10\%$ to $19\%$ \cite{adam2020clustering, endo2020estimating, sun2021transmission}.
Furthermore, \cite{lau2020characterizing} estimates that $20\%$ of all new infections are caused by the top-spreading $2\%$ of all infected.
In conclusion, while the shape of the offspring distribution remains unclear, there is ample evidence that it is heavily skewed, implying that super-spreaders are common.

In this work we test various offspring distributions to determine how much dispersion is needed to observe a fundamentally different behavior than the simpler deterministic SEIR model. We denote the corresponding random variable $X$, so that $P(X=n)$ is the probability that a given individual gives rise to $n$ new infections.
We then build an agent-based-model where each infected individual $i$ gives rise to $x_i$ new infections (in a fully susceptible population), where $x_i$ is drawn from the offspring distribution. The model is a stochastic extension of the classical deterministic SEIR-model, where S stands for susceptible, E for exposed, I for infective and R for recovered (further details in Section \ref{abm1}). We address the following questions:

\begin{itemize}
\item During a major outbreak, how does the (more realistic) AB-model deviate from a standard SEIR-model in a city of one million people?
\item At what threshold of ``fatness'' do we start to see unexpected behavior, such as an outbreak becoming self-extinct with high probability, if initiated with 100 import cases?
\end{itemize}
In short, our conclusion is that for all realistic offspring distributions, the answer to the first question is that the SEIR-model is a good approximation, and hence lacking detailed information about the offspring distribution is not a major drawback for modeling. Concerning the second question, we find that the probability of an outbreak becoming self-extinct increases with the fatness of the tail. Both these conclusions are in line with \textcolor{black}{well known theoretical results, see e.g.~Chapter 1 and 3 in \cite{diekmann2012mathematical} (albeit these are based on simpler models and hence are not necessarily applicable in our scenario)}. This is further discussed in Section \ref{priorwork}.

Note that $\mathbb{E}(X)=R_0$. For all distributions we chose parameters such that $R_0=1.3$. While most estimates indicate that $R_0$ for SARS-CoV-2 is higher, we picked $1.3$ since it is close to 1 and hence a random behavior is more likely than for higher values of $R_0$. We consider different distributions, and for each type we additionally pick parameters so that $p_{80}$ is either $20\%$, $10\%$ or $5\%$. The distributions are the following
\begin{itemize}
\item[(a)] A simple ``3-group distribution'' where each infected individual either spreads to 0, 1 or $n$ individuals, where $n=5$ ($p_{80}\approx 20\%$), $n=10$ ($p_{80}\approx 10\%$) and $n=20$ ($p_{80}\approx 5\%$).
\item[(b)] The negative binomial distribution with $k= 0.35$, $k= 0.11$ and $k= 0.05$.
\item[(c)] The (discretised) Pareto distribution with $\alpha=1.4$, $\alpha=1.2$ and $\alpha=1.13$.
\item[(d)] A truncated version of the above Pareto distribution to find out the role of extreme super-spreaders, defined as those responsible for more than 100 secondary infections.
\end{itemize}
See Figure \ref{fig1} for an illustration of (b) and (c), as well as Table \ref{tab:probTable}. As is clearly visible in the figure, the negative binomial distribution is in fact exponentially decaying, implying that high values of $X$ are improbable, in comparison with the Generalised Pareto distribution. For example, in the extreme case when $p_{80}=5\%$, we see that the two graphs cross at around 100 secondary cases. While it is virtually impossible to infect more than 1000 persons using the former model, the probability is 0.0076\% with the latter. Whether this is possible in reality or not is an important question, and plausible offspring distributions are hence discussed in the next subsection.

\begin{table}
\begin{center}
\caption*{Negative Binomial Distribution}
\begin{tabular}{ | c | c | c | c | c |}
\hline
 $k$ & $P(X=0)$ & $P(X > 10)$ & $P(X > 100)$ & $P(X > 1000)$ \\
 \hline
  0.05 &    85\%     & 3.6\% &     2.7e-02\% &  $<$ e-12\% \\
 \hline
0.11 &      75\%  &   3.4\% &  4.6e-04\%  & $<$ e-12\% \\
 \hline
0.35 &      58\%    & 1.4\% &  1.8e-10\%  & $<$ e-12\% \\
       \hline
\end{tabular}
\end{center}
\begin{center}
\caption*{Discretised Pareto Distribution}
\begin{tabular}{ | c | c | c | c | c |}
\hline
 $\alpha$ & $P(X=0)$ & $P(X > 10)$ & $P(X > 100)$ & $P(X > 1000)$ \\
 \hline
  1.13 &    86\%     & 1.2\% &     0.10\% &  7.6e-03\% \\
 \hline
1.2 &      81\%  &   1.4\% &  0.10\%  & 6.6e-03\% \\
 \hline
1.4 &      72\%    & 1.8\% &  0.087\%  & 3.5e-03\% \\
       \hline
\end{tabular}
\end{center}
\caption{Probability charts for some of the offspring distributions used below. Note the similarity of the first column, whereas the tail probabilities differ significantly.}
\label{tab:probTable}
\end{table}

\subsection{What is a reasonable offspring distribution for SARS-CoV-2?}\label{reasonable}

\begin{figure}
\begin{center}
\includegraphics[height=6.2 cm]{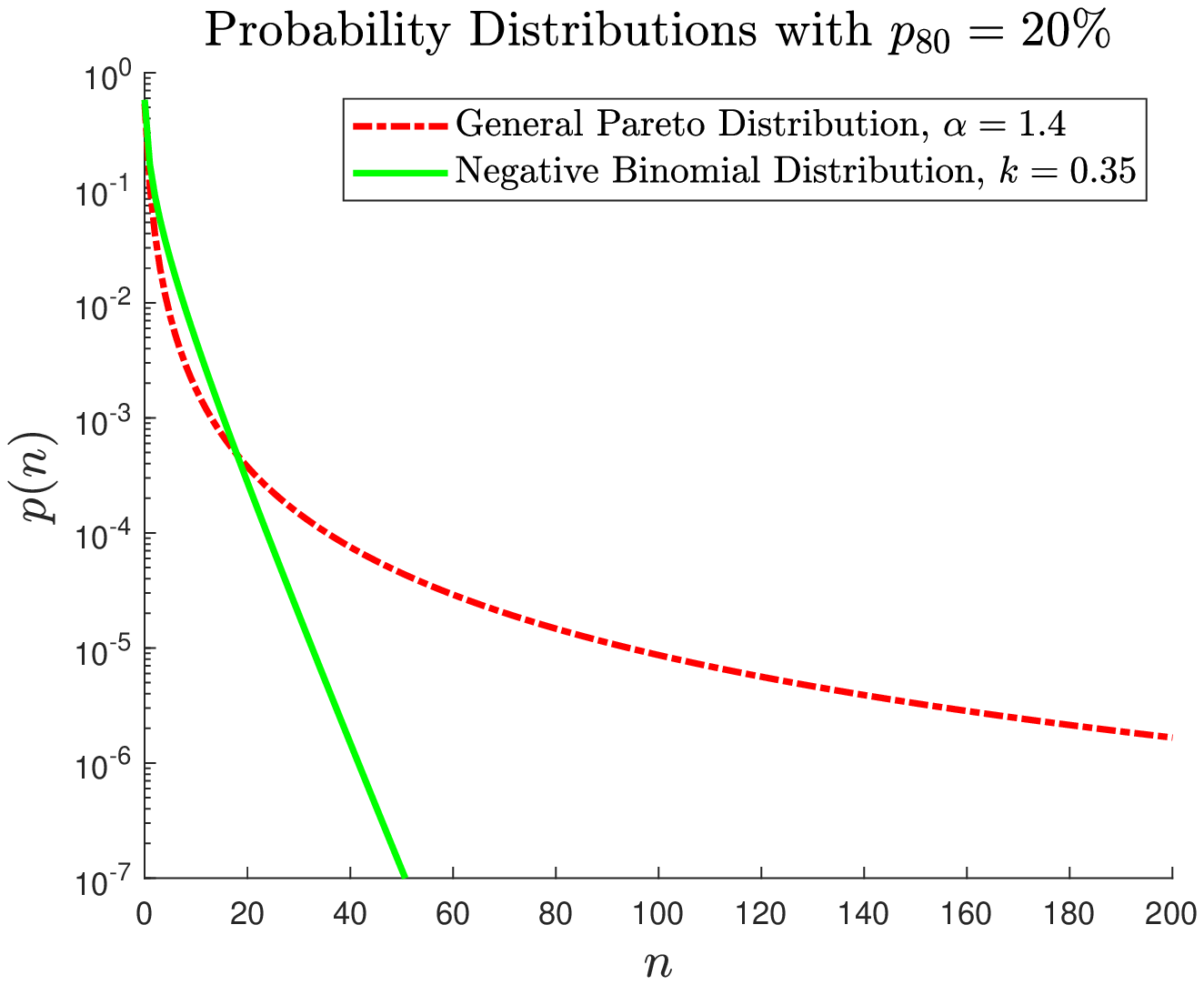}
\includegraphics[height=6.2 cm]{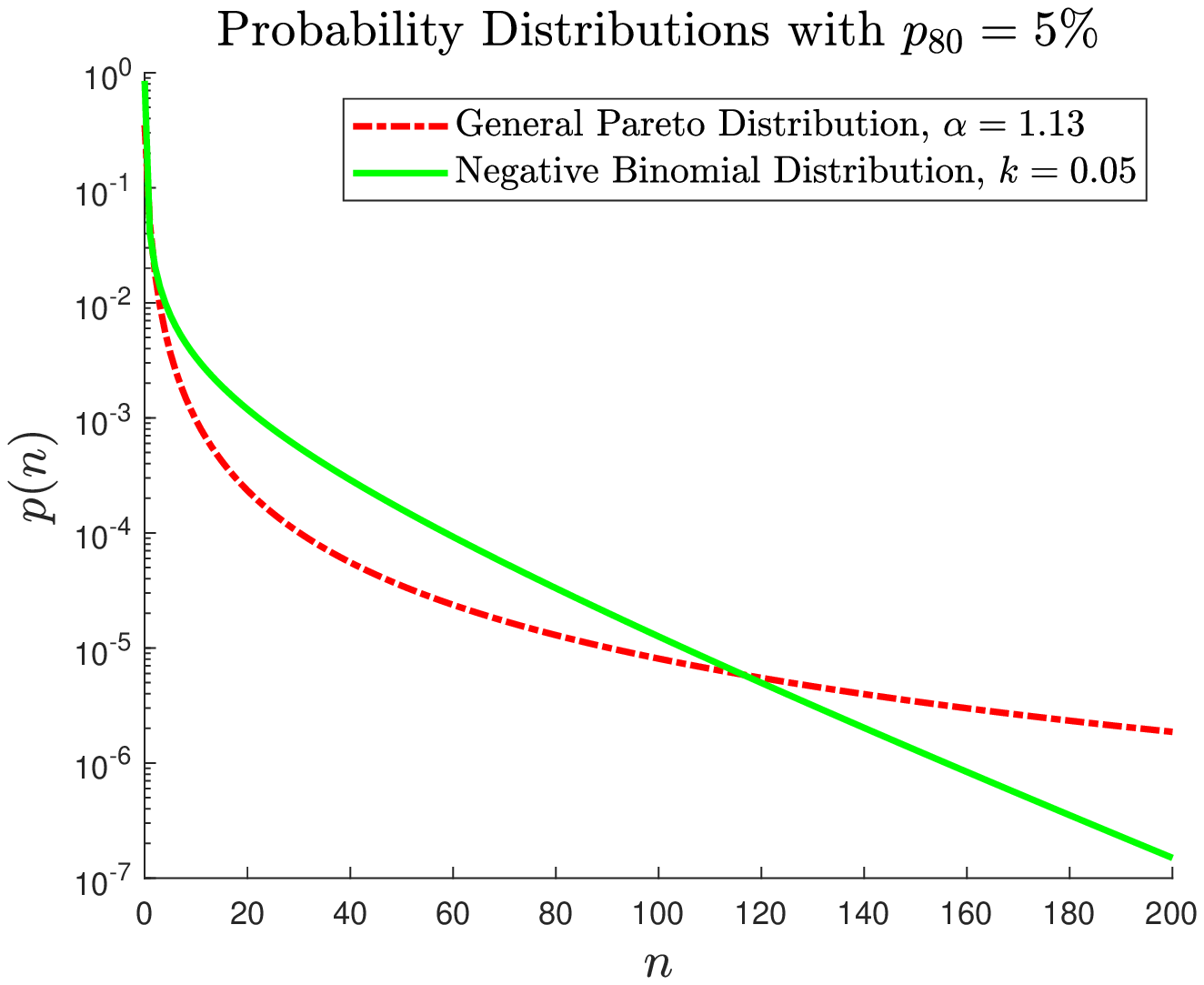}
\caption{Probability distribution functions of the Generalised Pareto Distribution and the Negative Binomial Distribution for various coefficients used in this study. }
\label{fig1}
\end{center}
\end{figure}

Data about the actual offspring distribution for SARS-CoV-2 comes from an analysis of contact-tracing studies. For example,  a large contact-tracing study from India  summarizes  the amount of secondary cases originating from 80000 index cases in a histogram \cite{laxminarayan2020epidemiology}. At first, it seems that the offspring distribution could be estimated from this by simply normalizing the total participants. Unfortunately, this gives rise to a distribution with an expectation value near $0.5$, implying that $R_0$ is well below one.
This could have many reasons. For example the data could indeed have been collected in a period when the effective reproductive rate ($R_e$) was below one, but a far more likely
explanation is that a sizable amount of SARS-CoV-2 infections were unidentified by contact tracing (transmission in restaurants, public transportation, etc).
Contact tracing studies are naturally biased in that they focus on \textit{known} contacts of the infective participant.

Still, the histogram from~\cite{laxminarayan2020epidemiology} gives valuable information, for example it shows that around $70\%$ of infected individuals gave rise to zero new infections, which means that a minority must give rise to the majority of new cases, which is one indication that the distribution could be ``fat-tailed'' or at least very skewed. But the fact that a majority of transmissions likely was unidentified calls for caution when interpreting the distribution. Another issue is that the study was carried out during strict NPI’s, which significantly limits super-spreader events and therefore should  affect the  tail of the offspring distribution.

The last point raises an important issue, namely that the offspring distribution clearly will depend not only on biological properties of the pathogen, but also on NPI’s and local factors such as household density and other socio-economic factors. Hence it is futile to talk about \textit{the} offspring distribution.

The study \cite{kremer2021quantifying} is an attempt at inferring the offspring distribution by fitting to three large contact tracing studies for SARS-CoV-2.
Unfortunately, all studies suffer the same drawbacks as noted above, reflected in that
they unexpectedly infer $R_e < 1$.
Since we are interested in how the offspring distribution looks when $R_0$ is above 1, it is questionable if these contact tracing studies can be used to derive any concrete information about the particular shape of the offspring distribution.

On the other hand, \cite{yang2021just} found that $2\%$ of test-positive cases (so called super-shedders) carries as much as $90\%$ of the total amount of SARS-CoV-2 virions circulating in a community.
This of course does not imply that $90\%$ of the spread is caused by these $2\%$, since the super-shedders cannot infect more people than they interact with. Still, it is an indication that the estimates of $p_{80}$ in the range $10-20\%$ (based on contact-tracing studies) may be overly conservative. For this reason, we included the case $p_{80}=5\%$ as well in our study.

In summary, there is currently no conclusive evidence for the correct value of $p_{80}$, let alone the shape of the offspring distribution. For these reasons we validate our conclusions with a number of fundamentally different distributions. We found that our conclusions are independent of which type of distribution is used, so the above mentioned uncertainties do not limit our findings.

To really test whether heavily skewed offspring distributions yield a different behavior, we will in particular run our model using a fat-tailed offspring distribution.
A fat-tailed distribution is one where the values $P(X=n)$ decay so slowly that the second moment (i.e.~the variance) does not exist. Neither of the first two distributions, (a) and (b), are fat tailed, which in practice means that the probability of hitting a high value of $x_i$ (the amount of secondary cases caused by individual $i$) is negligible. This is not the case for the Pareto distribution, which is the standard choice for modeling fat tails (see \ref{sec:pareto} for details). The significance of the parameter $\alpha$ in the Pareto distribution is that \begin{displaymath}
P(X\geq n)\approx n^{-\alpha},
\end{displaymath}
so the Pareto distribution is fat tailed whenever $\alpha<2$ and for $\alpha=1$ the tail becomes so fat that even the expectation value cease to exist. For values of $\alpha$ near one it is therefore no longer the case that we can neglect $P(X=x)$ for large values of $x$. For example, the probability of infecting more than $100$ persons is around $10^{-3}$ for all three Pareto distributions, and the probability of infecting more than $1000$ is above $ 10^{-5}$. In a city of one million, this means that on average 1000 persons have the capacity to generate more than 100 secondary cases, and between 10 and 100 individuals the capacity to generate more than 1000. This is in stark contrast with the Negative Binomial Distributions, where the probability of the latter is less than $10^{-14}$, meaning that we would need $10^8$ cities of one million each to find one such individual. This illustrates the vast difference between fat-tailed distributions and exponentially decaying ones.

But is it reasonable for one infected individual to infect more than a hundred, or a thousand? A super-market employee can easily interact with more than 100 customers each day, but
it is unclear whether such brief encounters are enough to infect most contacts. In any case, we have chosen to also include a truncated version of the Pareto-type distribution to better understand the effect of extreme super-spreaders when the tail is fat. We truncate at both $1000$ and at $100$, and discuss these results separately in Section \ref{sec:trunc}.

\begin{figure}[ht!]
    \begin{center}
    \includegraphics[height=8 cm]{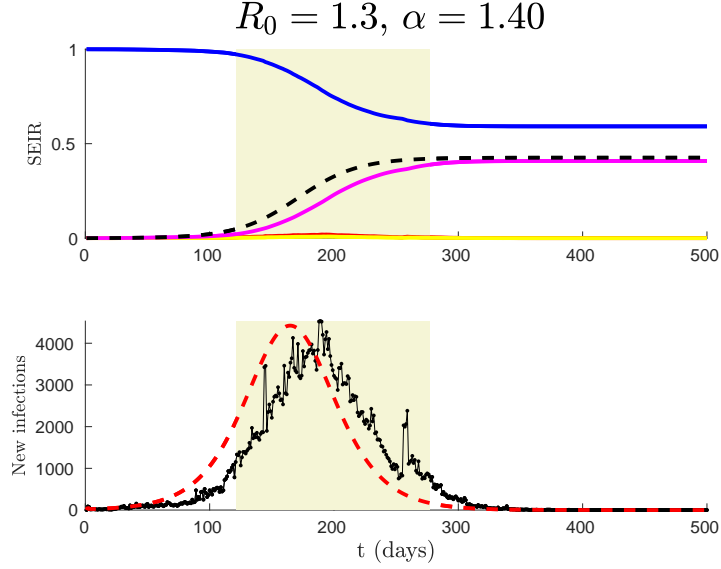}
    \caption{Result of an epidemic simulation with our AB-model and the classical SEIR, both using $R_0=1.3$. In the AB-model we used a Generalised Pareto Distribution with shape parameter $\alpha = 1.40$. In the upper figure, solid lines display the output of the deterministic SEIR; with blue for Susceptible, red for Exposed, yellow for Infective and pink for Recovered. The dashed black curve displays the Recovered group in one fixed trial using the AB-model. The final values of the solid pink and dashed black curves are the cumulative number of infections, known as the final size of the epidemic and denoted by $r(\infty)$, for the classical SEIR and the AB-model respectively. In the lower figure, the number of daily new infections for one trial of the AB-model is shown in black and for the SEIR model in dashed red. The grey area is considered to capture the ``epidemic wave'' (for the AB-model), corresponding to values of $R$ between $5\%$ and $95\%$ of $r(\infty)$. The width of this area then defines the ``epidemic wave time'' $T_{\textnormal{wave}}$.
    }
    \label{fig:example}
    \end{center}
\end{figure}

\subsection{Modeling details}\label{abm1}

We use the standard (deterministic) SEIR model
\begin{equation}\label{eq:seir}
\left\{\begin{array}{l}
        \dot s(t)=  -\alpha s(t)i(t)\\
         \dot e(t)=\alpha s(t)i(t)-\sigma e(t) \\
         \dot i(t)=\sigma e(t)-\gamma i(t) \\
         \dot r(t)=\gamma i(t)
       \end{array}
\right.\end{equation}
implemented with a step-size of one day (see \eqref{eq:seir2} in Appendix \ref{convergence} for details.)
%
%
We use the following parameter values; $T_{\textnormal{infective}}=2.1$, $T_{\textnormal{incubation}}=4.6$, $\alpha=R_0/T_{\textnormal{infective}}$, $\sigma=1/T_{\textnormal{incubation}}$ and $\gamma=1/T_{\textnormal{infective}}$, where $R_0$ is set to 1.3 in all trials.

We first sought to construct a model as close as possible to \eqref{eq:seir}, yet capturing the variability of disease transmission potential between individuals.
We model a population of $N$ individuals, each one in a state $z \in \{\mathrm{S}, \mathrm{E}, \mathrm{I}, \mathrm{R}\}$ respectively: susceptible, exposed, infected and recovered. Their state can vary once each day so that we denote by $z_{k,t}$ the state of individual $k$ at day $t$. Each individual has also a personal basic reproduction number $x_{k}$ drawn once from the offspring distribution.
Note that this is the number of people s/he is expected to infect, if s/he becomes infected, assuming that the population is fully susceptible. As immunity builds up, the actual outcome will be less. Hence $x_k$ should be interpreted as the total amount of \textit{potentially infections contacts} (i.e.~meetings with other individuals that lead to infection of the other in case s/he is susceptible). Moreover, since $T_{\textnormal{infective}}$ is the average time of infectiousness, this amounts to a rate of $x_k/T_{\textnormal{infectious}}=\gamma x_k$ daily such contacts. \textcolor{black}{We assume that the amount of infectives at any given moment is small compared to the total population, so that each person only meets at most one infective per day, and the total amount of potentially infectious contacts on any given day stays far below the total population, which certainly holds true in our simulations. }

At each time step (here each day) we compute $z_{k,t+1}$ as:
\begin{enumerate}
    \item Each infected agent adds its new daily contacts into a global pool of contacts, $N_{t} = \lfloor\gamma \sum_k x_{k} \delta_{z_{k,t}, \mathrm{I}}\rfloor$, where $\lfloor\cdot\rfloor$ denotes the integer part.
    \item $N_{t}$ individuals are drawn randomly among the global population. If they were susceptible, they become exposed.
    \item Each exposed agent has a $\sigma$ probability to become infected, $P(z_{k,t+1} = \mathrm{I} | z_{k,t} = \mathrm{E}) = \sigma$.
    \item Each infected agent has a $\gamma$ probability to be cured, and to become recovered, $P(z_{k,t+1} = \mathrm{R} | z_{k,t} = \mathrm{I}) = \gamma$.
\end{enumerate}

\textcolor{black}{See Figure \ref{fig:example} for a concrete realization in comparison with classical SEIR. Note that it takes a bit longer for the stochastic model to gain momentum, but that once this happens, the ensuing curve is well approximated by the deterministic SEIR model. For similar graphs using simpler models, see e.g.~Fig.~3.7 in \cite{diekmann2012mathematical}.}

The above model, which we refer to as AB-model 1, has a few principle drawbacks which we now describe. \textcolor{black}{Note that the total number of potentially infectious contacts of a particular individual equals $\gamma x_k $ times the amount of days that person remains in the infective compartment, which is determined by the process described in step 4 above. If we denote by $T$ the random variable describing the amount of days a person remains infective, then $T$ has an exponentially decaying PDF and a minor computation (see \ref{abm1appendix}) shows that $\mathbb{E}(T)=T_{\textnormal{infective}}$. The total amount of secondary infections is thus given by the random variable $\gamma X T$ and,
 since $X$ and $T$ are independent, we see that} \begin{equation}\label{andorra}\mathbb{E}(\gamma T X)=\gamma\mathbb{E}(T)\mathbb{E}(X)=\mathbb{E}(X)=R_0,\end{equation} \textcolor{black}{as it should be.
However, it is notable that the number of potentially infectious contacts is in the end \textit{not} given by the offspring distribution, in contrast to the initial motivation.} In fact, it is not even an integer. In addition, the assumption that an individual has the same probability of remaining sick each new day, independent of how long that individual has been sick already, is not very realistic. This drawback is not specific for the  AB-model 1 but pertains to the classical SEIR-model as well. In \cite{andersson2012stochastic} this is commented as follows: ``The assumption of an exponentially distributed infectious period is certainly not epidemiologically motivated, although with this assumption the mathematical analysis becomes much simpler.'' The authors of \cite{brauer2019mathematical} are on the same line; ``This assumption, while making the models and their analyses easier, is not biologically realistic for most infectious diseases'' (see Ch.~3.6). In section 4.5 and 4.6 of the same book they proceed to introduce more realistic systems of integral equations which keep track of the ``age of infection''. Such systems have a long history, see e.g.~\cite{hethcote1980integral} and \cite{lloyd2001destabilization}, where it is established that these models are more sensitive, and \cite{feng2007epidemiological} argues that the use of these more sophisticated methods can lead to different policy decisions. \textcolor{black}{Based on these facts, we argue that it is by no means evident that classical results concerning the final size of the pandemic and the probability of major outbreaks, as found in \cite{diekmann2012mathematical} Ch. 1 and 3, will apply in the modeling framework considered here.}

In order to make sure that our findings are not corrupted by the above mentioned drawbacks, we built a second AB-model that keeps track of the age of infection (called AB-model 2).
To do this we need the probability density function of the \textit{serial interval}, i.e.~the time from infection of an individual to the time s/he infects others. We obtained the PDF for this random variable, let us denote it $T_{SI}$, from the empirical study \cite{bi2020epidemiology}. We leave the details of the second model to the \ref{abm2}.
In conclusion, we found that the output of AB-models 1 and 2 are completely analogous.

\subsection{Relation to prior work}\label{priorwork}
As mentioned earlier, our main findings are in line with previously established results. A good overview source is \cite{diekmann2012mathematical}, in particular Chapters 1, 3 and 12. For instance, stochastic models for the spread of infectious diseases have previously been shown to behave like deterministic models when the population size tends to infinity~\cite{andersson2012stochastic,britton2019basic}. However, the stochastic ingredient in these models is that the time an infected individual remains infectious is random, whereas the degree of infectivity during the infectious period is assumed to be the same for all. Moreover, the proofs in this area are notoriously technical, we refer to \cite{armbruster2017elementary} for an overview and a (comparatively) ``elementary'' proof.

However, in this study we are mainly interested in the behavior during major outbreaks in a finite population with $N=10^6$, as well as in the probability that large outbreaks occur, and for these questions results of the above type could still be misguiding, since $10^6$ is far from $\infty$. Similarly, based on the theory for branching processes, it is well known that the probability of one infected individual giving rise to an outbreak that dies out by itself (minor outbreak), in an infinite population, is given by the smallest solution in $[0,1]$ to the equation \begin{equation}\label{genfunc}q=G(q):=\sum_{k=0}^\infty P(X=k)q^k\end{equation}
(cf.~\cite{diekmann2012mathematical} Chapter 1). It is easy to see that $G(0)=P(X=0)$, $G(1)=1$, $G'(1)=\mathbb{E}(X)=R_0$ and that $G$ is convex on $[0,1]$, so this equation has precisely one solution $\hat{q}$ in $[0,1)$. If we suppose that the 100 initial spreaders in our AB-model give rise to branching processes that do not intersect, one may suppose that the probability of our AB-model going self extinct is $\hat{q}^{100}$. That does not seem to be the case, for reasons that are unclear to us, although the value of $\hat{q}^{100}$ does provide a rough indication (see Section \ref{results}). \textcolor{black}{This underscores that one should be cautions in relying solely on theoretical predictions based on simplifying assumptions, and demonstrates the necessity of performing the actual modeling.} Moreover, it is not clear by looking at equation \eqref{genfunc} that a fatter tail necessarily leads to a higher probability of a minor outbreak, which is something that we shall establish experimentally.

The paper \cite{lloyd2005superspreading} investigates experimentally to what extent super-spreaders affects the probability of a major outbreak, finding that for small enough $k$-value erratic behavior appears even when assuming a negative binomial (thin tailed) distribution. For example, they calculate that the probability that the disease goes self-extinct is above 90\% if $R_0=1.5$ and the dispersion parameter $k$ equals 0.1 (for which $p_{80}\approx 10\%$). Our work differs from this in that we also consider fat-tail distributions, and that the epidemic is initiated with 100 rather than 1 individual, which is more realistic for import cases of COVID-19.

The paper \cite{wulkow2021prediction}  start with a detailed AB-model focusing on mobility patterns linked to Google data, however without taking variable infectivity into account. They conclude numerically that the output of such an advanced model, when applied to the city of Berlin, is well approximated by an 11-compartment ODE model (extended SEIR). This is thus in line with previous findings about stochastic models converging to deterministic models during large outbreaks~\cite{andersson2012stochastic,britton2019basic}, which we also confirm in the present contribution.

There are various papers studying numerically how various population heterogeneities affect the model curves, with the result that most heterogeneities such as variation in social activity or susceptibility, does have a damping effect on the overall spread \cite{britton2020mathematical,dolbeault2021social,gerasimov2021covid}. From this perspective, the finding that variable infectivity does not affect model output is a bit surprising. In contrast, variable infectivity has been shown to lead to larger outbreaks if it is positively correlated with variable susceptibility \cite{miller2007epidemic}, but this is not a topic we pursue here. Finally, the final size of an epidemic $r(\infty)$ can be computed with surprising accuracy, using only heuristic arguments (i.e.~without solving the SEIR), by solving an equation of the type \begin{equation}\label{finalsize}r_\infty=1-e^{-R_0r_{\infty}}\end{equation}
(see e.g.~\cite{diekmann2012mathematical}). In \cite{miller2012note}, this argument is extrapolated to show that variable infectivity (super-spreaders) does not effect the equation and hence not the final size. However, this conclusion is derived under the assumption that deterministic approximations apply, which is precisely what we want to establish in this paper. As a conclusion, we thus infer that the estimate $r_{\infty}$ from \eqref{finalsize} is a good approximation of the final size in our AB-model.

\section{Results}

\subsection{Theory}

Intuitively, once the population is large enough that the law of large numbers applies, differences in infectivity should average out. This is the essence of the below theorem.

\begin{theorem}\label{t0}
\textcolor{black}{Assume that $X$ has finite variance.} Let $T$ be a fixed integer and let $\bold{s},~\bold{e},~\bold{i}$ and $\bold{r}$ be random variables describing the fraction of susceptible, exposed, infected and recovered in a realization of the AB-model 1 from Section \ref{abm1}. In the limit as the population goes to infinity, these random variables converge in probability to a solution of (the discrete time version of) SEIR \eqref{eq:seir} for each $t$ in the time interval $\{1,\ldots,T\}$.
\end{theorem}
We refer to \ref{convergence} for a more stringent formulation.
As mentioned in Section \ref{priorwork}, classical proofs of similar results are notoriously difficult. Our framework differs from the above mentioned works in that we use a discrete time which, as noted in \cite{diekmann2021discrete}, is quite natural for many biological phenomena, in particular humans that need to sleep on a daily basis. The second key difference is the use of arbitrary offspring distributions.
We have not found any result establishing the same conclusion when the stochastic model also includes an individual offspring distribution, nor have we found a proof that works for SEIR, most references consider SIR. Our proof is interesting in its own right, since it avoids the ``Sellke construction'' \cite{diekmann2012mathematical} and the heavy machinery used by Ethier and Kurz \cite{ethier2009markov}. In fact, the proof is relatively simple in that it only relies on basic concepts from probability theory. It can be considered as an advanced version of the proof of the weak law of large numbers, ultimately boiling down to demonstrating that the variance goes to zero and then applying Chebyshev's inequality. The proof can easily be adapted to other similar models, and it is our hope that this can help improve the understanding of when deterministic models are a good approximation of stochastic counterparts.

\subsection{Evaluation criteria}

In order to answer the questions posed in Section \ref{novelties}, we need to decide on some key features as a basis for model comparison. In all our simulations, we model a city of $10^6$ people, starting with 100 infected individuals, or $0.01\%$ of the population, and set $R_0$ to $1.3$. We compare two epidemic endpoints, the ``final size of the epidemic'' $r(\infty)$, and the ``epidemic wave time'' $T_{\textnormal{wave}}$, described below.

In Figure \ref{fig:example} we show the dynamics of an epidemic computed with a classical SEIR model and with our AB-model. In the upper graph we have plotted the four SEIR time series, respectively Susceptible, Exposed, Infected and Recovered, as given by the classical deterministic model. In dashed black, we also display the Recovered time series $R(t)$ for our AB-model. The final values of the solid pink and dashed black curves are our first feature $r(\infty)$, the ratio of people that have been infected after the ``wave''. The grey area is what we consider the ``epidemic wave'', when the number of recovered $R(t)$ is between $5\%$ and $95\%$ of its final value $r(\infty)$. The width of this grey area corresponds to the epidemic wave time $T_{\textnormal{wave}}$. We remark that the deterministic value of $r(\infty)$ can be obtained by solving \eqref{finalsize}, whereas there is no closed formula regarding $T_{\textnormal{wave}}$, cf.~Ch 3 in \cite{diekmann2012mathematical}.

We study the effect that various type of offspring distributions have on the epidemic endpoints $r(\infty)$ and $T_{\textnormal{wave}}$, since these two values are the key parameters one is typically interested in predicting. In each trial we simulate $10^3$ epidemics and plot the histogram of the stochastic results together with the values computed with classical SEIR, see figures~\ref{fig:e1},  \ref{fig:e2} and~\ref{fig:e3}.

\begin{figure}[ht!]
        \centering
        \includegraphics[width=0.49\textwidth]{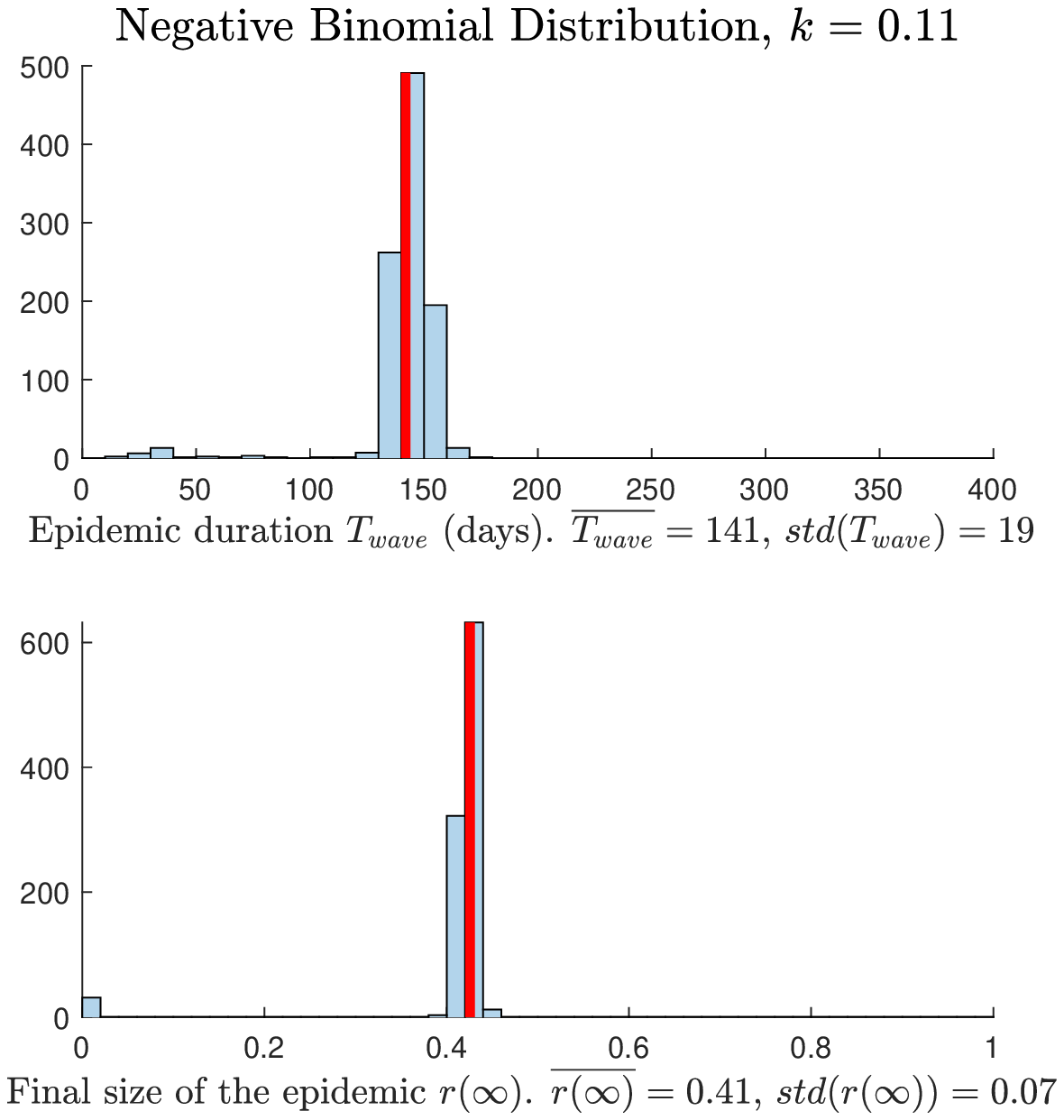}
        \centering
        \includegraphics[width=0.49\textwidth]{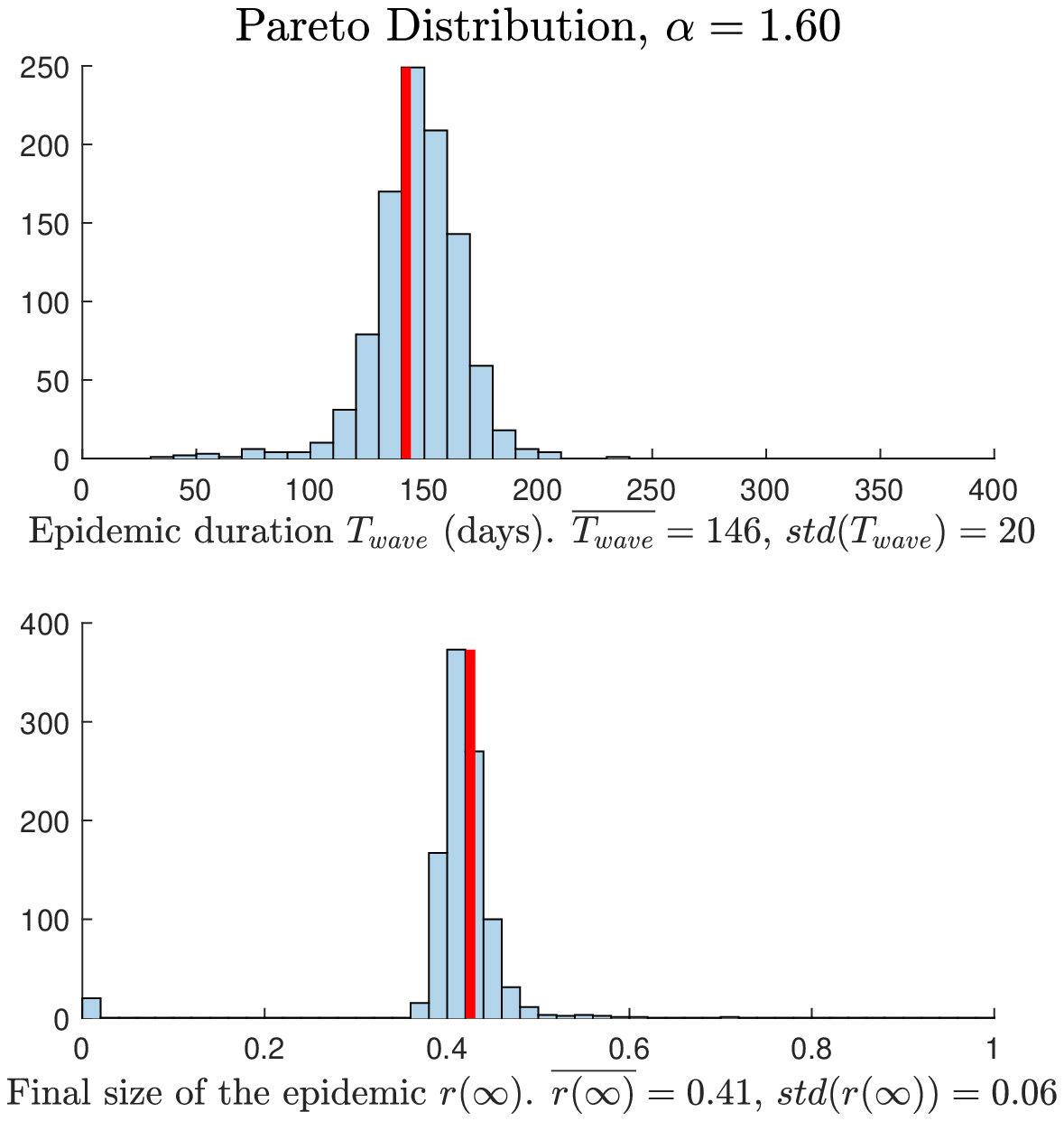}
        \caption{Case (1). The epidemic features are well predicted by the classical SEIR, shown by red solid lines. The histogram in blue shows the results of 1000 simulations with the AB-model.
				\textcolor{black}{Overbars, like $\overline{T_{\textnormal{wave}}}$, denotes averages, and $std$ denotes the standard deviation.} }
        \label{fig:e1}
\end{figure}

\subsection{Numerical Results}\label{results}

In figures~\ref{fig:e1}, \ref{fig:e2} and \ref{fig:e3}, we compare the epidemic features $r(\infty)$ and $T_{\textnormal{wave}}$ of the classical SEIR model, shown using vertical red lines, with the statistical results from a corresponding AB-model, displayed in the histograms.
The epidemic duration $T_{\textnormal{wave}}$ is plotted in the upper figures and the final size $r(\infty)$ in the lower figures. We have simulated epidemic dynamics with all the distributions described in Section~\ref{novelties}, (along with many other distributions and parameter values, which we have decided to omit from the presentation for brevity). The results fall into three overall categories, described below. 

\begin{itemize}
    \item[(1)] For distributions with thinner tails, the AB-model gives the same result as the classical SEIR. The distribution of spreaders is too thin-tailed to influence the epidemic features $r(\infty)$ and $T_{\textnormal{wave}}$ and the AB-model is still mostly deterministic. Occasionally an epidemic goes self-extinct in its wake (minor outbreak), but when this does not happen the epidemic will terminate at a value close to the deterministic $r(\infty)$ in a time frame also well predicted by the deterministic $T_{\textnormal{wave}}$.

        For the Negative Binomial Distribution this happens for $k=0.11$ ($p_{80}=10\%$) and above. Using a Generalised Pareto distribution, we fall into this category with $\alpha \gtrsim 1.6$ ($p_{80}\approx 25\%$), see figure~\ref{fig:e1}.
    \item[(2)] When the distribution tail becomes thicker, minor outbreaks become more common. However, if this does not happen, then its characteristics of the major outbreak are still well predicted by the classical SEIR model. We exemplify this in Figure \ref{fig:e2} using the Negative Binomial Distribution with $k=0.05$ ($p_{80}=5\%$) as well as the Generalised Pareto distribution with $\alpha=1.4$, ($p_{80}=20\%$). As a general trend, the deterministic estimate of $r(\infty)$ has less variance than that for $T_{\textnormal{wave}}$.
    \item[(3)] When we get closer to very fat tailed distributions, with $\alpha \lesssim 1.2$ in the Generalised Pareto distribution (i.e.~$p_{80}\lesssim 10\%$), most of the trials lead to minor outbreaks, see figure~\ref{fig:e3} where around $900$ simulations out of $1000$ obtain $r(\infty) \approx 0$.
\end{itemize}

\begin{figure}[ht!]
        \centering
        \includegraphics[width=0.49\textwidth]{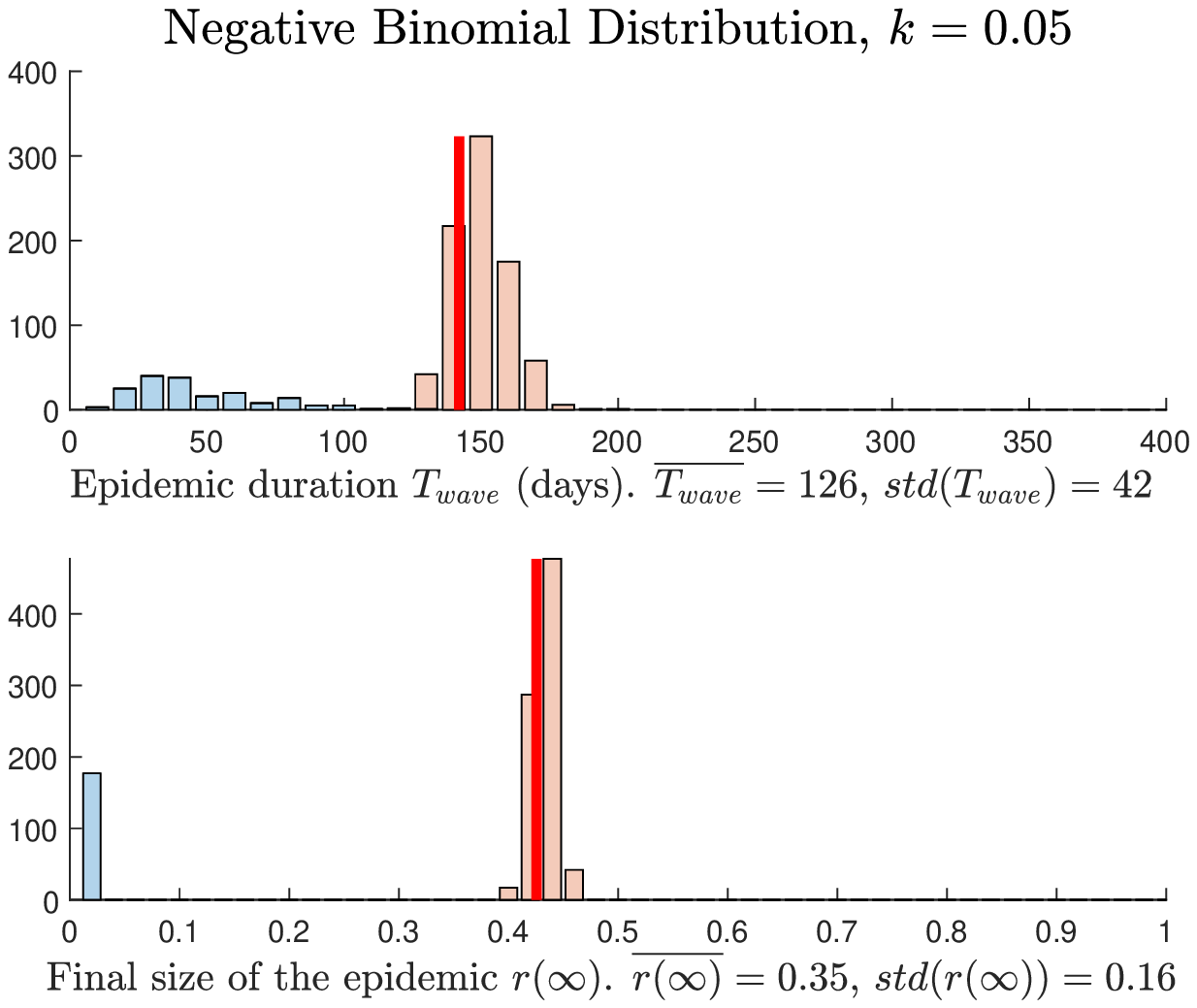}
        \centering
        \includegraphics[width=0.49\textwidth]{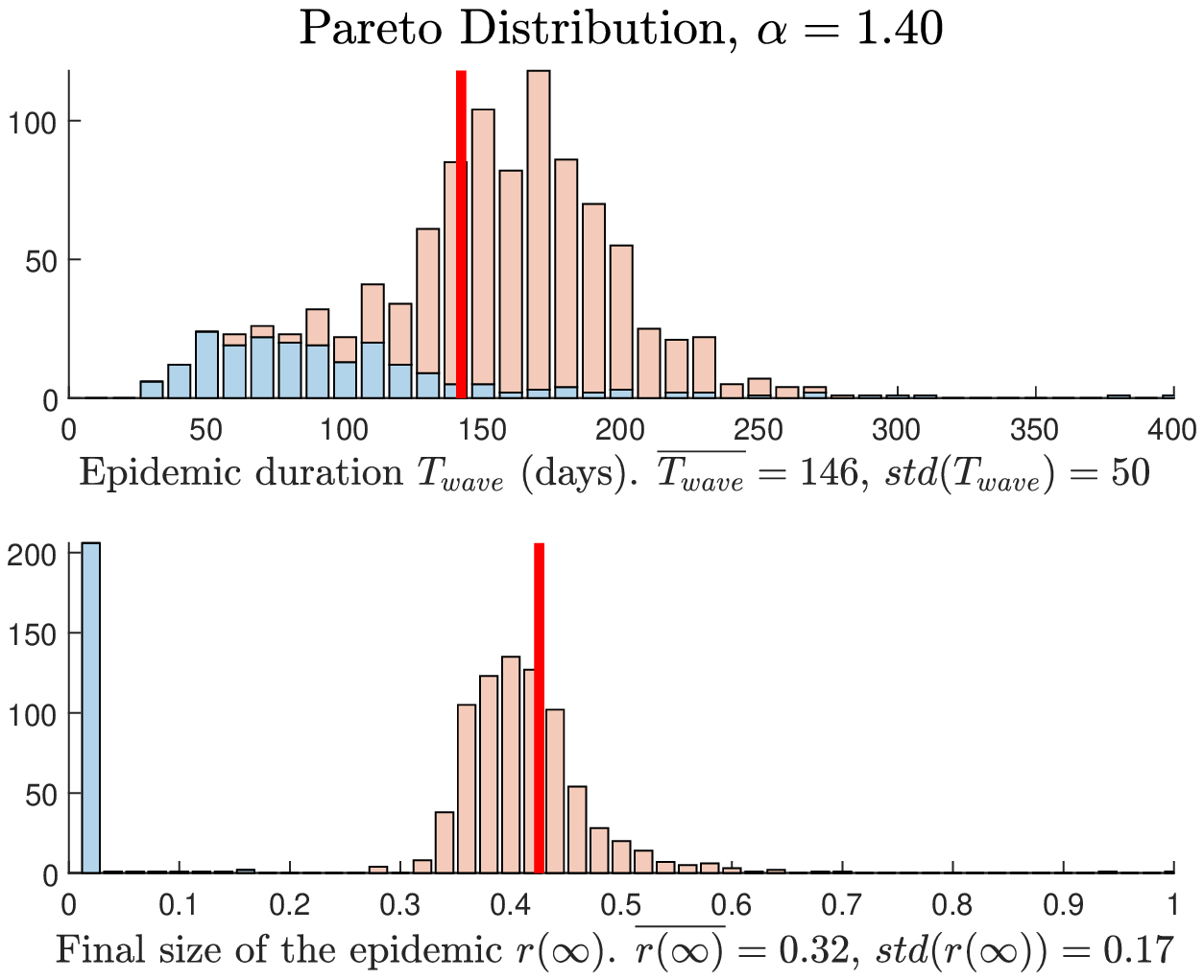}
        \caption{Case (2). A dichotomy appears, either the epidemic becomes self-extinct, or its features are relatively well predicted by the classical SEIR.
				We differentiate between the two kinds of outcomes by color coding, where blue refers to epidemics going self-extinct and orange refers to epidemics with $r(\infty)$ near the deterministic value $r_{\infty}$. }
        \label{fig:e2}
\end{figure}

\begin{wrapfigure}{R}{0.5\textwidth}
 	\begin{center}
       \includegraphics[width=0.49\textwidth]{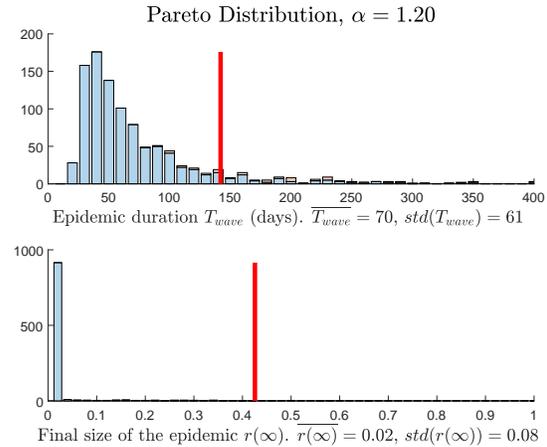}
 	\end{center}
        \caption{Case (3). Most epidemics become self-extinct after infecting only a fraction of the population (minor outbreak). }
        \label{fig:e3}
 \end{wrapfigure}

For the 3-groups distribution, the simulations done with $p_{80}=20\%$, $10\%$ and $5\%$ fall into the case (1). However, for $p_{80}=5\%$, around $80$ epidemics out of $1000$ yield only a minor outbreak, so we see the beginning of a random behavior like in the case (2). As mentioned in Section \ref{priorwork}, the probability of a minor outbreak can be estimated theoretically by solving \eqref{genfunc}. In our experience, this provides a rough estimate of what is observed in the model. For example, with the parameters as in Figure \ref{fig:e2}, the theoretical probability of a minor outbreak was 9\% (based on equation \eqref{genfunc}) whereas we observed roughly 17\% when using the binomial distribution (left). However, for the Pareto distribution (right) the order was reversed, with a theoretical value of 37\% and an observed value of around 23\%. Our modeling in general suggests that thicker tails for the offspring distribution (keeping $R_0$ fixed) lead to higher probability of a minor outbreak. This is not clear to see from equation \eqref{genfunc}, and we have also not verified that this holds theoretically.

To sum up, we see a dichotomy where the outcome of the AB-model either is fairly well predicted by the deterministic SEIR model, or the epidemic becomes self-extinct at an early stage. The parameter $p_{80}$ is not enough to determine how the AB-model will behave; for the 3-groups distribution we need a very low value of $p_{80}$ to end up in Case (2) or (3), for the Negative Binomial Distribution this happens already at $p_{80}\approx 10\%$, and for the Generalised Pareto distribution at $p_{80}\approx 20\%$.

In conclusion, for the deterministic SEIR model, replacing a constant basic reproduction number with a distribution of spreaders has a marginal effect on the epidemic dynamic once it has clearly started (irrespective of whether we are in Case (1), (2) or (3)).
However, it has a major effect on the probability of the epidemic to start or not, at least for distributions that are sufficiently skewed. We have shown that this probability is decreasing with the thickness of the distribution's tail, but we were not able to formulate a definitive measure to predict it.

\subsection{Truncating the Pareto distribution}\label{sec:trunc}

As expected, we found the Generalised Pareto Distribution to produce the most erratic behavior of the distributions tested. In this section we come back to the question of whether the Pareto distribution is suitable for describing the spread of an infectious virus such as SARS-CoV-2. In particular, what is the relevance of extreme events where one individual gives rise to an enormous number of secondary cases? To answer this, we truncated the distribution at 1000 (as well as at 100) and ran the experiment with the AB-model again. In practice, this means that anytime a value $x_i$ larger than 1000 is drawn, we simply redefine $x_i$ to be 1000. The results are displayed in Figure \ref{fig:e4}. Comparing with the untruncated simulations in Figure \ref{fig:e2} (right), we see that the more we truncate, the more the average of $r(\infty)$ gets shifted to the left, and the more $T_{\textnormal{wave}}$ gets shifted to the right.
However, observe that when truncating we change $R_0=\mathbb{E}(X)$, so the observed shifts are completely in line with the general observations of the previous section. More precisely, in both cases we see a dichotomy where the epidemic either becomes self-extinct, or leads to a major outburst whose magnitude and duration are well predicted by the deterministic SEIR using the updated $R_0$-value.

\begin{figure}[ht!]
        \centering
        \includegraphics[width=0.49\textwidth]{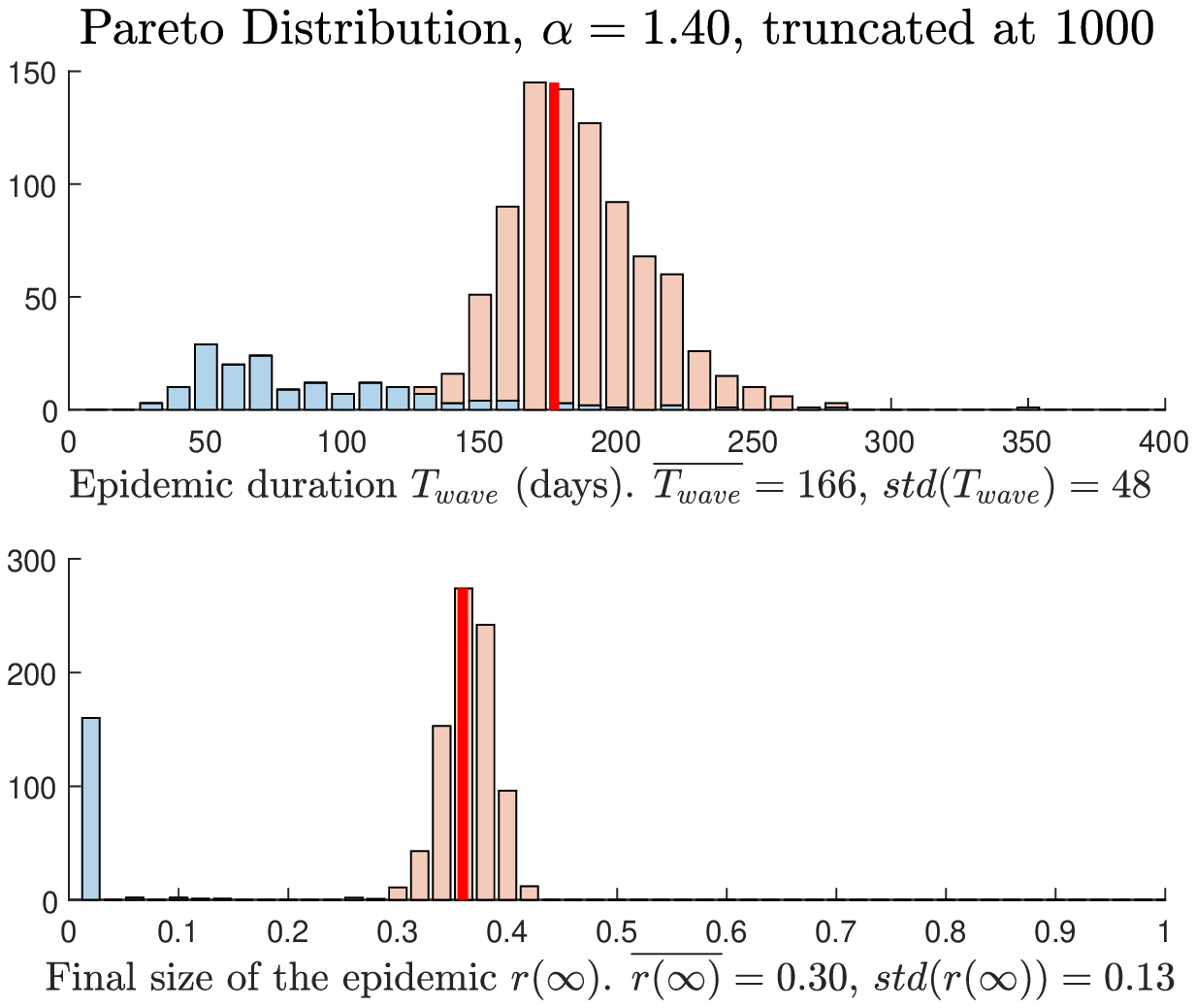}
        \centering
        \includegraphics[width=0.49\textwidth]{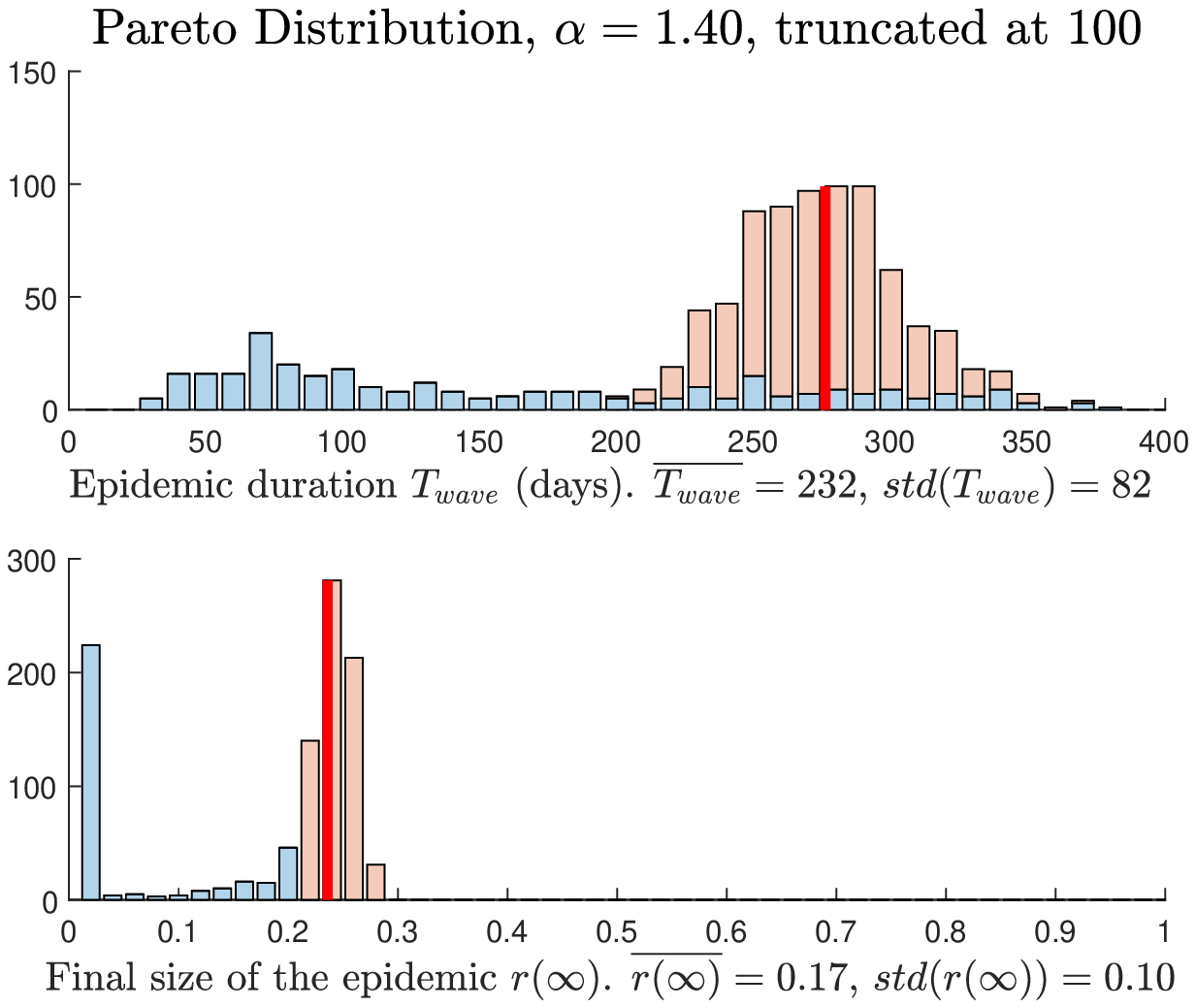}
        \caption{By truncating we shift the expectation and hence the $R_0$-value. However, we observe the same dichotomy as described in the previous section.}
        \label{fig:e4}
\end{figure}

\section{Discussion}

The COVID-19 pandemic has been notoriously hard to predict with deterministic models for spread of infectious diseases. We have investigated numerically whether this could be related to the existence of super-spreaders, but found that this does not seem to be the case.
We tested the difference between the classical SEIR model and an AB-model using several fundamentally different offspring distributions, from fat-tailed ones to exponentially decaying ones to distributions with finite support. For each given distribution we kept $R_0$ fixed and measured skewness by the parameter $p_{80}$. Irrespective of distribution used, we identified a similar qualitative pattern;
For large enough $p_{80}$ the behavior of the AB-model is always well predicted by the deterministic SEIR, with somewhat larger fluctuations for the predicted $T_{\textnormal{wave}}$ than the final size $r(\infty)$. This is followed by a transition zone where the above conclusion is still true whenever the model leads to a substantial outbreak, and when this does not happen we only see a minor outbreak that becomes self-extinct. For small enough $p_{80}$, the epidemics almost never start, hence this case is too extreme and not relevant for understanding the spread of SARS-CoV-2.

Our numerical observations are in sync with related works and theoretical predictions, both from simpler models and our own proof that the AB-model 1 converges to SEIR as the population size approaches infinity.
As discussed in Section \ref{reasonable}, there is great uncertainty about the actual shape of the offspring distribution, which also likely is country specific. Moreover, the precise value of $p_{80}$ remains unknown, due to the large amount of spread that is not registered in contact tracing studies. A strength of the conclusions of this paper is thus that they are independent of the type of distribution as well as the value of $p_{80}$.
In addition, we have also confirmed numerically that given an overdispersed offspring-distribution, the relative chance of the epidemic going self-extinct is increased even if $R_0>1$, with the likelihood of self-extinction being negatively correlated to $R_0$ and positively correlated to a low $k$-value or, more generally, a fat tail.

\section*{Statements and Declarations}
\subsection*{Funding}
Authors F.~R.~and M.~\"{O}.~received financial support from Carl Tryggers foundation and RR-ORU-2021/2022. Remaining authors received no support for this work.
\subsection*{Competing interests}
All authors declare no competing interests.
\subsection*{Author contributions}
All authors contributed to the study conception and design. B.~K.~W.~came up with the original research question, M.~C.~contributed with model design and analysis, M.~\"{O}. and F.~R.~did the actual modeling. We thank S. von Blixen for proofreading.
\subsection*{Ethics approval}
Ethics approval is not needed for studies on mathematical modeling.
\subsection*{Consent to participate and publish}
The study included no human subjects.
\subsection*{Data availability statement}
The code used to produce the figures are available from the corresponding author upon reasonable request.


\newpage

\section{Appendix}\label{sec:appendi}

\subsection{The Pareto distribution}\label{sec:pareto}

Here we describe in detail the generalised Pareto distribution used in the main text. This distribution has two parameters: scale $\sigma$, and shape $\alpha$ (the location $\mu$ is set to 0). The shape parameter is traditionally called $\xi$ and is related to $\alpha$ by $\alpha=1/\xi$, but we have decided to work with $\alpha$ since this is easier to interpret for non-specialists.
The continuous PDF (Probability Distribution Function) of the generalised Pareto distribution is given by

\begin{equation}\label{eq:p_alpha}
    f_{\alpha,\sigma}(x)=\frac{1}{\sigma}\left(1+\frac{x}{\alpha\sigma}\right)^{-\textcolor{black}{(\alpha+1)}},\quad x>0,
\end{equation}
and the expectation of a corresponding random variable $X_{\alpha, \sigma}$ is then

\begin{equation}\label{eq:hy}
    R_0=\mathbb{E}(X_{\alpha,\sigma})=\frac{\alpha\sigma}{\alpha-1}.
\end{equation}
The distribution becomes ``fat-tailed'' when its variance becomes infinite, here when $\alpha<2$. The scale parameter $\sigma=\sigma(\alpha)$ will always be chosen so that $R_0=1.3$. We have gathered in Table~\ref{tab:probTableGP} several points of the Cumulative Distribution Function with different $p_{80}$ and shape parameter $\alpha$.

\begin{table}[!h]
    \begin{center}
        \begin{tabular}{ | c | c | c | c | c | c |}
        \hline
         $p_{80}$ & $\alpha$  & $P(X \geq 1)$ & $P(X \geq 10)$ & $P(X \geq 100)$ & $P(X \geq 1000)$ \\
         \hline
          5 & 1.125 &    10.6\%    & 0.11\% &  8.16e-4\%  & 6.14e-6\% \\
         \hline
        10 & 1.197 &      14.2\%  &   0.14\% &  9.41e-4\%  & 6.01e-6\% \\
         \hline
        10.2 & 1.2 &      14.3\%  &   0.142\% &  9.43e-4\%  & 5.98e-6\% \\
         \hline
        19.6 & 1.4 &      20.5\%    & 0.196\% &  8.77e-4\%  & 3.53e-6\% \\
         \hline
        20 & 1.415 &      20.8\%    & 0.20\% &  8.16e-4\%  & 3.36e-6\% \\
         \hline
        24.8 & 1.6 &      24.0\%    & 0.222\% &  6.65e-4\%  & 1.70e-6\% \\
         \hline
        34.2 & 2.5 &      30.1\%    & 0.225\% &  1.24e-04\%  & 4.17e-8\% \\
         \hline
         42 & 10 &      34.7\% &   0.096\%  &  1.42e-9\%  &  4.23e-20\% \\
           \hline
        \end{tabular}
    \end{center}
\caption{Probability chart of the Generalised Pareto Distributions for various values of $p_{a}$.}
\label{tab:probTableGP}
\end{table}

\textcolor{black}{The continuous cumulative distribution function $F_{\alpha,\textnormal{cont}}(x)=\int_0^x f_{\alpha,\sigma(\alpha)}(y)dy$ is then given by $1-\left(1+\frac{x}{\alpha\sigma}\right)^{-\alpha}$.}
We discretise this to the set of integers by simply setting $F_{\alpha}(n)=F_{\alpha,\textnormal{cont}}(n)$, which leads to the discrete probability density function
\begin{equation}
\label{eq:p_alpha_disc}
p_{\alpha}(n)=F_{\alpha}(n+1)-F_{\alpha}(n).
\end{equation}
The distinction between $F_{\alpha}$ and $F_{\alpha,\textnormal{cont}}$ is a major change which e.g.~completely alters the formula~\eqref{eq:hy}. In the discretised case there is no closed form expression for the expectation, and hence we need to use a small numerical routine to find an appropriate value of $\sigma$ (i.e.~such that $R_0=\mathbb{E}(X_{\alpha})=\sum_{n=0}^\infty np_{\alpha}(n)$ holds).

As argued in the main text, it seems unrealistic that one individual can give rise to more than, say $100$ or $1000$ new cases, so to test how much the tail beyond a certain number $N_0$ affects the performance of the AB-model, we introduce a second truncated distribution $F_{\alpha,trunc}$ by simply redefining $F_{\alpha}$ so that $F_{\alpha,trunc}(n)=1$ for all $n>N_0$, i.e., $p_\alpha(n>N_0)=0$ according to~(\ref{eq:p_alpha_disc}).
This interference will lower the value of the expectation (i.e.~$R_0$), which of course could be countered by adjusting $\sigma$, but in the examples displayed in the main text we chose not to do so, in order to better see the influence of the truncation on a concrete distribution (see Figure \ref{fig:e4}). In conclusion, we see that the truncation, whether at 100 or at 1000, does not alter the overall behavior of the AB-model in terms of the key features $T_{\textnormal{wave}}$ and $r(\infty)$, except for a shift that is completely predicted by the updated value of $R_0.$

\subsection{Agent Based model 1}\label{abm1appendix}
\textcolor{black}{
Since the main features of this model is described in Section \ref{abm1}, we shall not repeat it here. It remains to prove the claim that $\mathbb{E}(T)=T_{\textnormal{infective}}$ which was used in equation~(\ref{andorra}). Based on Step 4 in the algorithm, we see that the probability of being cured after one day only is $\gamma$. The probability of being cured after two days is $(1-\gamma)\gamma$, three days $(1-\gamma)^2\gamma$ and so on. Thus $$\mathbb{E}(T)=\sum_{t=0}^\infty t\gamma(1-\gamma)^{t-1}=\gamma\frac{d}{d\gamma}\left(-\sum_{t=0}^\infty (1-\gamma)^{t}\right)=\gamma\frac{d}{d\gamma} \left(-\frac{1}{1-(1-\gamma)}\right)=\gamma\frac{d}{d\gamma}\left( -\gamma^{-1}\right)=\gamma^{-1}=T_{\textnormal{infectious}},$$
as desired.}

\subsection{Agent Based model 2}\label{abm2}

\begin{figure}[ht!]
        \centering
        \includegraphics[width=0.49\textwidth]{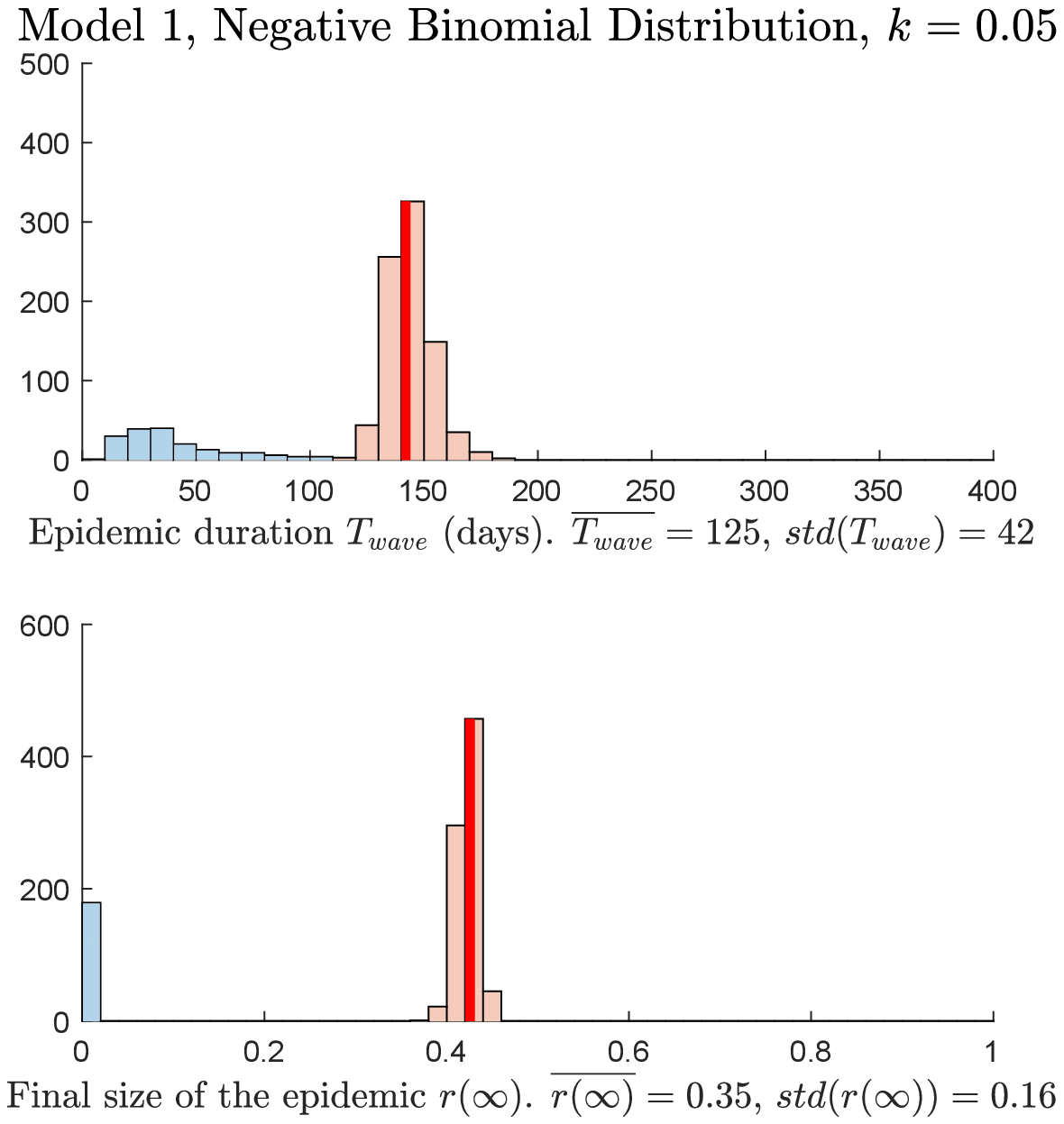}
        \centering
        \includegraphics[width=0.49\textwidth]{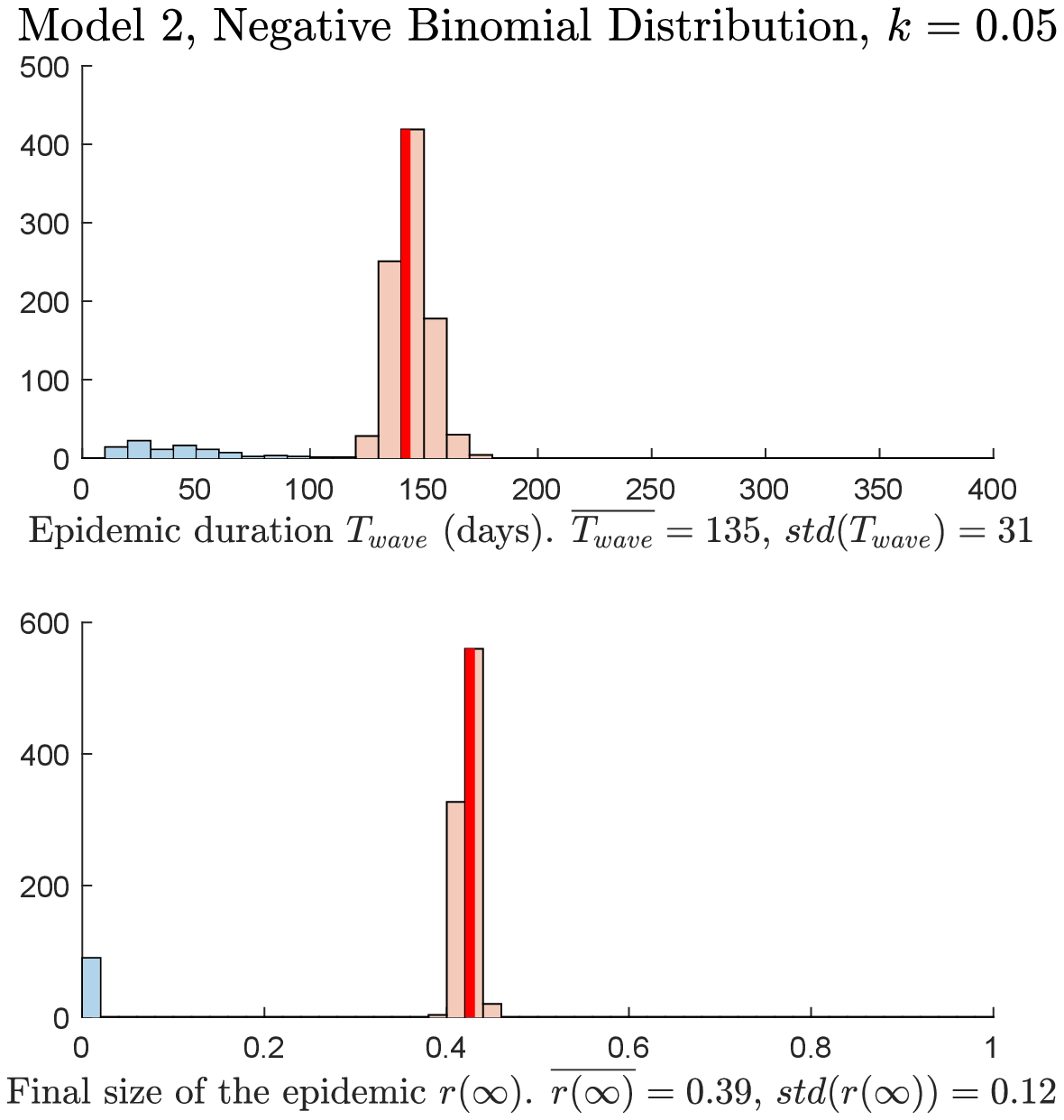}
        \caption{Comparisons of the different AB-models with the same parameters.}
        \label{fig:SnuR}
\end{figure}

In this section we describe our variation of the agent based model presented in Section \ref{abm1}, with the key difference that this model takes the age of infection into account. \textcolor{black}{To this end, we need an estimate of the PDF of the serial interval $T_{SI}$, i.e.~a function $p_{T_{SI}}$ such that $p_{T_{SI}}(t)$ is the probability that onward transmission of the virus happens on day $t$ after a given infected individual got infected. We used an estimate of $T_{SI}$ from \cite{bi2020epidemiology}, where in particular we set $p_{T_{SI}}(0)=0$.}
As in Section \ref{abm1} we model a population of $N$ individuals, each individual $k$ having a state-variable $z_{k,t}$ which keeps track of the current status of that individual. However, instead of $s$, $e$, $i$ and $r$, we now introduce the states $\mathrm{S}$ (meaning Susceptible as before), $\nu$ (for newly infected) and $\mathrm{R}$ where $\mathrm{R}$ stands for ``Removed'' as opposed to ``Recovered''. The reason for this semantic change is that
with this construction a person can only be in the $\nu$-group one day, and is automatically moved to $\mathrm{R}$ on the subsequent day, and hence the $\mathrm{R}$-group will also contain people that are still sick and even pre-symptomatic.

We introduce a function $n$ on the non-negative integers, where $n(t)$ is the amount of potentially infectious contacts on the given day $t$. This function takes the role of the variable $N_{t}$ in AB-model 1, whose value only depended on the amount of infectious people on day $t$. The value $n(t)$, in contrast, depends on the entire times series prior to day $t$.
\textcolor{black}{At each time step $t$ we compute $z_{k,t+1}$ and update $n$ as follows:
\begin{enumerate}
    \item For each infected agent, i.e.~$k$ with $z_{k,t}=\nu$, we draw a number $x_k$ from the offspring distribution.
    \item We then draw $x_k$ integers $\tau_1,\ldots,\tau_{x_k}$ from the $T_{SI}$ distribution. Note that these numbers by construction are non-zero.
    \item For each such value $\tau_j$ we add $+1$ to the function $n$ on day $\tau_j+t$. Formally, $$n(\tau):=n(\tau)+\sum_{j=1}^{x_k}\delta_{\tau,\tau_j+t}.$$
    \item All individuals in state $\nu$ are given state $\mathrm{R}$.
    \item $n(t)$ individuals are drawn randomly among the global population. Among these, those that are in state $\mathrm{S}$ are moved to state $\nu$ on day $t+1$. Formally, if individual $m$ is drawn and $z_{m,t}=\mathrm{S}$, we set $z_{m,t+1}=\nu$.
\end{enumerate}}

The above model is a stochastic version of the integral equation type models found e.g.~in \cite{hethcote1980integral} and \cite{lloyd2001destabilization} or Chapter~4 of \cite{brauer2019mathematical}.
Before we can compare this model 2 with the previous one, we need to determine suitable values for $T_{\textnormal{incubation}}$ and $T_{\textnormal{infective}}$ (used in AB-model 1) in accordance with the PDF for $T_{SI}$.
A standard analysis (see e.g., the supplementary material of~\cite{carlsson2021mathematical}) yields that the average time for the serial interval in the SEIR model becomes $T_{\textnormal{incubation}}+T_{\textnormal{infective}}$, and hence it is important for a comparison to run model 1 with values such that $$T_{\textnormal{incubation}}+T_{\textnormal{infective}}=\mathbb{E}(T_{SI}).$$
When this condition is met, we have found that there is only marginal difference in the histogram output of model~1 and~2. In Figure~\ref{fig:SnuR} we show the output of the two models using the Negative Binomial Distribution with $k=0.05$. The only notable difference in the figure is that there seems to be a somewhat higher probability that the epidemic becomes self-extinct when using model~1 (this observation seems to be a general trend). However, the overall behavior follows the same pattern. In particular, we see the dichotomy between the epidemic either becoming self-extinct or its key features $T_{\textnormal{wave}}$ and $r(\infty)$ being well predicted by the classical SEIR. Hence the conclusions of the main text will be identical independent of the particular stochastic model.

\subsection{Convergence of AB-model 1 to discrete SEIR}\label{convergence}
Given a fixed (small) $\epsilon>0$ let $s_d$, $e_d$, $i_d$ and $r_d$ be the deterministic solution to the discrete version of the equation system \eqref{eq:seir} implemented with
a step size of one day, i.e.
\begin{equation}\label{eq:seir2}
\left\{\begin{array}{l}
        s_d(t+1)= s_d(t)  -\alpha s_d(t)i_d(t)\\
        e_d(t+1)= e_d(t) + \alpha s_d(t)i_d(t)-\sigma e_d(t) \\
        i_d(t+1)=i_d(t) + \sigma e_d(t)-\gamma i_d(t) \\
        r_d(t+1)=r_d(t) + \gamma i_d(t)
       \end{array}
\right. , \end{equation}
with initial condition $e_d(0)=r_d(0)=0$, $i_d(0)=\epsilon$ and $s_d(0)=1-\epsilon$. We remark that the difference between solutions to this implementation and the corresponding ODE \eqref{eq:seir} is marginal and completely irrelevant for the tests in this paper. For an interesting account of how to discretize ODE's like \eqref{eq:seir}, in particular with longer step-lengths, see \cite{diekmann2021discrete}.

In the AB-model 1 (recall Section \ref{abm1}), let $$\{\bold{z}_{k,t}\}_{k=1}^N$$ be random variables describing the state $\{S,E,I,R\}$ of individual $k$ on day $t$. Set $\bold{S}(t)=\sum_{k=1}^N\delta_{\bold{z}_{k,t},\mathrm{S}}$ and define $\bold{E}$, $\bold{I}$ and $\bold{R}$ analogously. These are then random variables (with values in $\{0,1,\ldots,N\}$) that implicitly depend on the $N$ random variables $\{\bold{z}_{k,t}\}_{k=1}^N$ (we use bold letters for random variables). Finally, introduce $\bold{s}=\bold{S}/N$ and define $\bold{e}$, $\bold{i}$ and $\bold{r}$ analogously. The key question posed in this paper is whether the deterministic model \eqref{eq:seir2} is a good approximation of the average behavior of the random AB model 1, with $N=10^6$ and $\epsilon=100/10^6$, for a variety of offspring distributions that may reflect the actual spread pattern of SARS-CoV-2.

In this section, we prove mathematically that the random model converges in probability to the deterministic one, as $N\rightarrow\infty$. Note that this result does not answer the key questions posed, namely what happens when $N$ is fixed at a finite value realistic for a medium-large city. In the below result, we initialize the AB-model 1 by randomly selecting $\lfloor\epsilon N\rfloor$ infectives on day $t=0$, letting the remaining variables $z_{k,t}$ equal $\mathrm{S}$, (where $\lfloor\epsilon N\rfloor$ denotes the integer part of $\epsilon N$).

\textcolor{black}{It is important that $\epsilon$ is chosen small enough that $i_d(t)<1/\alpha$ for all values of $t$. For ``reasonable'' parameter values, this is always the case, but still this assumption is necessary since otherwise we will get negative values for $s_d$ in \eqref{eq:seir2}, which is absurd. Another technicality (which significantly complicates the proof) is that the number $N_t$, representing the amount of potentially infectious contacts in step 1 of AB-model 1, in theory can exceed $N$, in which case the subsequent steps are not well defined. In all our trials with this model $N_t$ was always way below $N$, but still we need to add the rule that $N_t$ is set to $N$ if this happens. With these modifications we have:}

\begin{theorem}\label{t1}
Assume that $\bold{X}$ has finite variance. Let $T$ be a fixed integer \textcolor{black}{and assume that $i_d(t)<1/\alpha$ holds for all $t=1,\ldots,N$}. The random variables $\bold{s},~\bold{e},~\bold{i}$ and $\bold{r}$, initialized as above, converge in probability to $s_d,~e_d,~i_d$ and $r_d$ for all $t$ in the time interval $\{1,\ldots,T\}$. More precisely, fix any $\delta>0$ and let $\Omega_N$ be the set of realizations $(s,e,i,r)$ such that $|s(t)-s_d(t)|<\delta$, $|e(t)-e_d(t)|<\delta,$ $|i(t)-i_d(t)|<\delta$ and $|r(t)-r_d(t)|<\delta$ holds for all $ t=1,\ldots,T.$
Then $\lim_{N\rightarrow\infty}\mathbb{P}(\Omega_N)=1.$
\end{theorem}
To make the proof more transparent, we begin by proving a few preliminary results. Let $\{\bold{X}_k\}_{k=1}^N$ be i.i.d.~random variables whose distribution function is that of the offspring distribution, and introduce the random variable
\begin{equation}\label{i9}
\bold{N}(t)=\min\left(\left\lfloor\gamma\sum_{k=1}^N \bold{X}_k\delta_{\bold{z}_{k,t},\mathrm{I}}\right\rfloor,N\right).
\end{equation}
Note that the number $N_t$, chosen in accordance with step 1 in AB-model 1, modified to avoid too large numbers as described earlier, is drawn from the PDF of $\bold{N}(t)$. Let $\bold{n}(t)=\bold{N}(t)/N$, let $f_N$ be the function $f_N(x)=\min\left(\frac{1}{N}\lfloor Nx\rfloor,1\right)$ and set $f(x)=\min(x,1)$.
Finally let $\bold{p}(t)$ be the random variable \begin{equation}\label{santamarta}\bold{p}(t)=\frac{\gamma}{N}\sum_{k=1}^N \bold{X}_k\delta_{\bold{z}_{k,t},\mathrm{I}}\end{equation}
and note that $\bold{n}(t)=f_N(\bold{p}(t))$.
Our overall goal is to establish that $\bold{n}(t)$ converges in probability whenever $\bold{i}(t)$ does. We begin with some basic lemmas.

\begin{lemma}\label{l2}
Let $\bold{A}_N$ and $\bold{B}_N$ be uniformly bounded real random variables on the same probability space (which may vary with $N$), and assume that they converge in probability to the values $A$ and $B$ respectively, as $N\rightarrow\infty$. Let $f$ be any continuous function on $\mathbb{R}^2$.
Then
$$
\lim_{N\rightarrow\infty}\mathbb{E}(f(\bold{A}_N,\bold{B}_N))=f(A,B),
$$
and $\lim_{N\rightarrow\infty}\mathbb{V}ar(f(\bold{A}_N,\bold{B}_N))=0.$
\end{lemma}
\begin{proof}
Let $\mu$ be the joint probability measure for $(\bold{A}_N,\bold{B}_N)$ on $\mathbb{R}^2$, and denote by $C_R$ the square $\{(a,b):~\max(|a|,|b|)\leq R\}$. By the assumption of  uniform boundedness we can pick an $R$ such that
$
\mu_N(\mathbb{R}^2\setminus C_R)=0
$
for all $N$. Fix any $\delta>0$, let
$$
S_\delta=\{(a,b):~|a-A|<\delta,~|b-B|<\delta\},
$$
and note that $\lim_{N\rightarrow\infty}\mu_N(S_\delta)=1$ by assumption. Then \begin{align*}&|\mathbb{E}(f(\bold{A}_N,\bold{B}_N))-f(A,B)|=\left|\int f(a,b)~d\mu_N-f(A,B)\right|=\left|\int_{C_R} f(a,b)-f(A,B)~d\mu_N\right|\leq\\ &\sup_{(a,b)\in S_\delta}\left|f(a,b)-f(A,B)\right|\mu_N(S_\delta)+\sup_{(a,b)\in C_R}\left|f(a,b)-f(A,B)\right|\mu_N(C_R\setminus S_\delta) \end{align*}
which implies that the limit as $N\rightarrow\infty$ is bounded by $\sup_{(a,b)\in S_\delta}\left|f(a,b)-f(A,B)\right|$. But this number can be made arbitrarily small, by the continuity of $f$, which establishes the first part of the lemma.

For the second part, note that $\mathbb{V}ar(f(\bold{A}_N,\bold{B}_N))=\mathbb{E}\Big(\big(f(\bold{A}_N,\bold{B}_N)\big)^2\Big)-\Big(\mathbb{E}\big(f(\bold{A}_N,\bold{B}_N)\big)\Big)^2$
and both terms converge to $\big(f(A,B)\big)^2$ by the first part of the lemma.
\end{proof}

In particular, we have the following consequence.

\begin{lemma}\label{l1}
Let $\bold{A}_N$ be a sequence of uniformly bounded real random variables on some probability space (which may vary with $N$), and let $A$ be a fixed number. Then $(\bold{A}_N)_{N=1}^\infty$ converges in probability $A$, as $N\rightarrow\infty$, if and only if $$\lim_{N\rightarrow\infty}\mathbb{E}[\bold{A}_N]=A$$ and $\lim_{N\rightarrow\infty}\mathbb{V}ar[\bold{A}_N]=0$.
\end{lemma}
\begin{proof}
The ``only if'' part follows immediately from the previous lemma. The reverse is a simple application of Chebyshev's inequality.  More precisely, let $\mu_N$ be the probability measure for $\bold{A}_N$ on $\mathbb{R}$ and note that $$\mathbb{V}ar(\bold{A}_N)=\int (x-\mathbb{E}[\bold{A}_N])^2 d\mu_N\geq \epsilon^2\mu_{N}((A_N-\epsilon,A_N+\epsilon)),$$
where $A_N=\mathbb{E}(\bold{A}_N)$ and $\epsilon>0$ is any fixed number. Now if $\lim_{N\rightarrow\infty}A_N=A$ and the variance converges to 0, the above inequality clearly implies that $$\lim_{N\rightarrow\infty}\mu_{N}((A-\epsilon,A+\epsilon))=0,$$
as desired.

\end{proof}

With this we are ready to prove our first key result.

\begin{proposition}\label{p1} If $\bold{i}(t)$ converges in probability to some value $i^*$, then $\bold{p}(t)$ converges in probability $\alpha i^*$.\end{proposition}

\begin{proof}

We have \begin{equation}\label{calor}\mathbb{E}[\bold{p}(t)|\bold{i}(t)=i]=\frac{\gamma}{N}iN\mathbb{E}(\bold{X})=\gamma i R_0=\alpha i\end{equation} by \eqref{andorra} and the parameter settings following \eqref{eq:seir}, where we also used that there are precisely $iN$ non-zero copies of $\bold{X}$ in the sum \eqref{santamarta}. Thus the law of total expectation implies that
$$\mathbb{E}(\bold{p}(t))=\mathbb{E}\big[\mathbb{E}[\bold{p}(t)|\bold{i}(t)]\big]=\alpha \mathbb{E}(\bold{i}(t)),$$ which in turn yields that $\lim_{N\rightarrow\infty}\mathbb{E}(\bold{p}(t))=\alpha i^*$.
Similarly, note that $$\mathbb{E}[\bold{p}^2(t)|\bold{i}(t)=i]=\big((iN)^2-(iN)\big)(\frac{\gamma}{N}\mathbb{E}(\bold{X}))^2+(iN)\frac{\gamma^2}{N^2}\mathbb{E}(\bold{X}^2)
.$$ Combining this with \eqref{calor}, we get \begin{align*}&\mathbb{V}ar(\bold{p}(t)|\bold{i}(t)=i)=\mathbb{E}[\bold{p}^2(t)|\bold{i}(t)=i]-(\mathbb{E}[\bold{p}(t)|\bold{i}(t)=i])^2=\\&
\big(i^2-\frac{i}{N}\big)({\gamma}\mathbb{E}(\bold{X}))^2+i\frac{\gamma^2}{N}\mathbb{E}(\bold{X}^2)
-({\gamma}i\mathbb{E}(\bold{X}))^2\\&=\frac{i{\gamma^2}}{N}\left(\mathbb{E}(\bold{X}^2)-(\mathbb{E}(\bold{X}))^2\right)=\frac{i{\gamma^2}}{N}\mathbb{V}ar(\bold{X})\end{align*}
since the first and last term on the middle line cancel each other. By the law of total variance we thus get  $$\mathbb{V}ar(\bold{p}(t))=\mathbb{E}[\mathbb{V}ar(\bold{p}(t)|\bold{i}(t)=i)]+\mathbb{V}ar(\mathbb{E}[\bold{p}(t)|\bold{i}(t)=i])=
\frac{\mathbb{E}(\bold{i}(t)){\gamma^2}}{N}\mathbb{V}ar(\bold{X})+\gamma\mathbb{V}ar(\bold{i}(t))\mathbb{E}(\bold{X}).$$
Recall that both $\mathbb{E}(\bold{i}(t))$ and $\mathbb{V}ar(\bold{i}(t))$ implicitly depend on $N$. By the assumption on $\bold{i}(t)$ and Lemma \ref{l1} we have  $\lim_{N\rightarrow\infty}\mathbb{E}(\bold{i}(t))=i^*$ so the first term (on the right) goes to zero due to the division by $N$. The second term also goes to zero, since $\lim_{N\rightarrow\infty}\mathbb{V}ar(\bold{i}(t))=0$, again by Lemma \ref{l1}.

We have shown that $\lim_{N\rightarrow\infty}\mathbb{E}(\bold{p}(t))=\alpha i^*$ and $\lim_{N\rightarrow\infty}\mathbb{V}ar(\bold{p}(t))=0$, so $\bold{p}(t)$ converges in probability to $\alpha i^*$ by Chebyshev's inequality (the details are the same as in the proof of Lemma \ref{l1}, which holds even though $\bold{p}(t)$ is not bounded).
\end{proof}

In order to convert this into a similar statement about $\bold{n}$, we first need another lemma.

\begin{lemma}\label{l3}
Let $\bold{A}_N$ be a sequence of real random variables on some probability space (which may vary with $N$), and assume that $(\bold{A}_N)_{N=1}^\infty$ converges in probability to the value $A$, as $N\rightarrow\infty$. Let $f_N$ be a sequence of uniformly bounded functions on $\mathbb{R}$ that converge uniformly to some continuous function $f$. Then $
(f_N(\bold{A}_N))_{N=1}^{\infty}
$ converges in probability to $f(A)$.
\end{lemma}
\begin{proof}
Let $\mu_N$ be the probability measure for $\bold{A}_N$ on $\mathbb{R}$ and pick any $a,b$ such that $a<A<b$. Note that $$\mathbb{E}(f_N(\bold{A}_N)-f(A))=\int_{\mathbb{R}} f_N-f(A)~d\mu_N=\int_{(-\infty,a]} f_N-f(A)~d\mu_N+\int_{(a,b)} f_N-f(A)~d\mu_N+\int_{[b,\infty)} f_N-f(A)~d\mu_N.$$
The first and last contribution are bounded in modulus by $\sup\{|f_N(x)-f(A)|:~x,N\}\mathbb{P}(\bold{A}_N\not \in (a,b))$, which goes to 0 by assumption on $\bold{A}_N$. The central term is bounded by $$\sup\{|f_N(x)-f(A)|:~a<x<b,~N\}\mu_N((a,b))<\sup\{|f_N(x)-f(A)|:~a<x<b,~N\},$$ which in the limit as $N\rightarrow\infty$ equals $\sup\{|f(x)-f(A)|:~a<x<b\}$. By the continuity of $f$ we see that this bound can be made arbitrarily small, upon choosing $a$ and $b$ close to $f(A)$. This proves that $$\lim_{N\rightarrow\infty}\mathbb{E}(f_N(\bold{A}_N))=f(A),$$
and the fact that the variance converges to 0 now follows exactly as in the previous lemma.

\end{proof}

We can finally establish the following:

\begin{proposition}\label{p2} If $\bold{i}(t)$ converges in probability to some value $i^*$, then $\bold{n}(t)$ converges in probability to $\min(\alpha i^*,1)$.\end{proposition}

\begin{proof}
Recall that with $f_N(x)=\min\left(\frac{1}{N}\lfloor Nx\rfloor,1\right)$ we have $\bold{n}=f_N(\bold{p})$, and that $(f_N)_{N=1}^\infty$ converges uniformly to $\min(x,1)$. The desired conclusion is now an immediate application of Proposition \ref{p1} and the previous lemma.
\end{proof}

Armed with this, we can prove the main result as follows:

\begin{proof}[Proof of Theorem \ref{t1}]
Denote $(\bold{s},\bold{e},\bold{i},\bold{r})$ by $\bold{y}$. Let $\Omega_{N,t}$ be the set of realizations such that the four previous inequalities hold on day $t$, so that $\Omega_N=\cap_{t=1}^T \Omega_{N,t}.$ We then have $$\mathbb{P}(\{\bold{y}\in \Omega_N^c\})=\mathbb{P}(\cup_{t=1}^T\{\bold{y}\in \Omega_{N,t}^c\})\leq \sum_{t=1}^T \mathbb{P}(\{\bold{y}\in \Omega_{N,t}^c\}),$$
(where $\Omega^c$ denotes the complement of $\Omega$), by which it follows that it suffices to prove
\begin{equation}\label{trs}\lim_{N\rightarrow\infty}\mathbb{P}(\{\bold{y}\in \Omega_{N,t}\})=1,
\end{equation}
for a fixed but arbitrary $t.$

This will now be established using induction.
It is certainly true for $t=0$, because the discrepancy in $y_d(0)$ and $y_r(0)$ (for any realization $y_r$ of $\bold{y}$) caused by taking the integer part of $\epsilon N$ is less than $1/N$, and hence $\mathbb{P}(\{\bold{y}\in \Omega_{N,0}\})=1$ whenever $N>1/\delta$.

Now we assume that \eqref{trs} holds for a given $t$.
In order to establish \eqref{trs} for $t:=t+1$ it suffices to prove that each of the four probabilities $$\mathbb{P}(|\bold{s}(t+1)-s_d(t+1)|<\delta),~\mathbb{P}(|\bold{e}(t+1)-e_d(t+1)|<\delta),~\mathbb{P}(|\bold{i}(t+1)-i_d(t_0+1)|<\delta)$$ and  $\mathbb{P}(|\bold{r}(t+1)-r_d(t+1)|<\delta)$ converge to 1 as $N\rightarrow \infty$.
Let $\bold{M}_{se}(t)$ be the random variable which denotes the amount of people who, on day $t$, leave $\mathrm{S}$ and enter $\mathrm{E}$. Define $\bold{M}_{ei}$ and $\bold{M}_{ir}$ analogously, and let $\bold{m}_{se}$, $\bold{m}_{ei}$ and $\bold{m}_{ir}$ be the corresponding normalized variables.

Our first goal will be to establish that
\begin{equation}\label{i5}
\lim_{N\rightarrow\infty}\mathbb{P}(|\bold{s}(t+1)-s_d(t+1)|<\delta)=1
\end{equation}
for all $\delta>0$, which by Lemma \ref{l1} is equivalent with proving that $$\lim_{N\rightarrow\infty}\mathbb{E}(\bold{s}(t+1))=s_d(t+1)=s_d(t)-\alpha s_d(t)i_d(t)$$ and that $\lim_{N\rightarrow\infty}\mathbb{V}ar(\bold{s}(t+1))=0$. In particular, the inductive assumption implies that $\lim_{N\rightarrow\infty}\mathbb{E}(\bold{s}(t))=s_d(t)$ and $\lim_{N\rightarrow\infty}\mathbb{V}ar(\bold{s}(t))=0.$
Since clearly $\bold{s}(t+1)=\bold{s}(t)-\bold{m}_{se}(t+1)$, we are done if we establish that \begin{equation}\label{bogota}
\mathbb{E}[\bold{m}_{se}(t+1)]=\alpha s_d(t)i_d(t)\end{equation} and $\lim_{N\rightarrow\infty}\mathbb{V}ar(\bold{m}_{se}(t+1))=0$. (To see this, note that the expectation is linear and  that  $$\mathbb{V}ar(\bold{s}(t)-\bold{m}_{se}(t+1))\leq 2\mathbb{V}ar(\bold{s}(t))+2\mathbb{V}ar(\bold{m}_{se}(t+1))$$ which follows from the basic inequality $|2xy|\leq x^2+y^2$).

We begin by proving \eqref{bogota}.
It is standard that the probability for $\bold{M}_{se}(t+1)=M$, given that $\bold{S}(t)=S$ and $\bold{N}(t)=N_t$, is given by the hypergeometric distribution (sampling without replacement)
\begin{equation}\label{pdef}
\mathbb{P}(\bold{M}_{se}(t+1)=M|~\bold{S}(t)=S,\bold{N}(t)=N_t)=\frac{\binom{S}{M}\binom{N-S}{N_t-M}}{\binom{N}{N_t}},
\end{equation}
for which we have the well known formulas \begin{equation}\label{pdeff}\mathbb{E}[\bold{M}_{se}(t+1)|~\bold{S}(t)=S,~\bold{N}(t)=N_t]=N_t\frac{S}{N}\end{equation} and $\mathbb{V}ar (\bold{M}_{se}(t+1)|~\bold{S}(t)=S,\bold{N}(t)=N_t)=N_t\frac{S}{N}\frac{N-S}{N}\frac{N-N_t}{N-1}$.
Now, by the law of total probability we have
$$\mathbb{E}[\bold{M}_{se}(t+1)]=\mathbb{E}\left[\mathbb{E}[\bold{M}_{se}(t+1)|~\bold{S}(t),~\bold{N}(t)]\right]=\mathbb{E}\left[\bold{N}(t)\frac{\bold{S}(t)}{N}\right]$$
which upon division by $N$ yields
$$\mathbb{E}[\bold{m}_{se}(t+1)]=\mathbb{E}\left[\bold{n}(t){\bold{s}(t)}\right].$$

By the induction hypothesis we know that $\bold{s}(t)$ converges in probability to $s_d(t)$, so the desired identity \eqref{bogota} follows by Lemma \ref{l2} and the fact that
$\bold{n}(t)$ converges to $\alpha i_d(t)$ in measure, as established in Proposition \ref{p2} (and the assumption that $i_d(t)<1/\alpha$).

We now turn to the identity $\lim_{N\rightarrow\infty}\mathbb{V}ar(\bold{m}_{se}(t+1))=0$. By \eqref{pdeff} and the subsequent formula for the variance we have
$\mathbb{E}[\bold{m}_{se}(t+1)|~\bold{s}(t)=s,~\bold{n}(t)=n]=ns$ and $$\mathbb{V}ar (\bold{m}_{se}(t+1)|~\bold{s}(t)=s,\bold{n}(t)=n)=ns (1-s)(1-n)\frac{1}{N-1}.$$
The law of total variance thus implies $$\mathbb{V}ar (\bold{m}_{se}(t+1))=\frac{1}{N-1}\mathbb{E}\left[\bold{n}\bold{s} (1-\bold{s})(1-\bold{n})\right]+\mathbb{V}ar[\bold{n}\bold{s}].$$
The above expectation converges to a finite number as $N\rightarrow\infty$ by Lemma \ref{l2} and Proposition \ref{p2}, hence the first term goes to 0 due to the factor $\frac{1}{N-1}$. By Lemma \ref{l2} again, the variance of the second term converges to 0 as well.

This completes the proof that \eqref{i5} holds. The proofs of the corresponding identities for $\bold{e}$, $\bold{i}$ and $\bold{r}$ are simpler, so we omit the details.
\end{proof}

\end{document}